\UseRawInputEncoding
\documentclass[11pt,reqno]{amsart}
\allowdisplaybreaks[4]
\usepackage{graphicx} 
\usepackage{subfigure}
\usepackage{epsfig}
\usepackage{amssymb}
\usepackage{amsmath,xypic}
\usepackage{cite,color}
\numberwithin{equation}{section}
\newtheorem{theorem}{Theorem}[section]
\newtheorem{definition}[theorem]{Definition}
\newtheorem{proposition}[theorem]{Proposition}
\newtheorem{corollary}[theorem]{Corollary}

\newtheorem{lemma}[theorem]{Lemma}
\newtheorem{remark}[theorem]{Remark}
\newcommand{\be}{\begin{equation}}
\newcommand{\ee}{\end{equation}}
\newcommand{\bea}{\begin{eqnarray}}
\newcommand{\eea}{\end{eqnarray}}
\newcommand{\ba}{\begin{array}}
	\newcommand{\ea}{\end{array}}
\newcommand{\bean}{\begin{eqnarray*}}
	\newcommand{\eean}{\end{eqnarray*}}

\begin{document}
\title{ Polynomial tau-functions of the symplectic KP, orthogonal KP and BUC hierarchies}

\author{ Denghui Li and Zhaowen Yan$^*$}
\dedicatory {School of Mathematical Sciences, Inner Mongolia University, \\ Hohhot, Inner Mongolia 010021,  P.\ R.\ China}
\thanks{*Corresponding author. Email: yanzw@imu.edu.cn}

\begin{abstract}

 This paper is concerned with the construction of the polynomial tau-functions of the symplectic KP (SKP), orthogonal KP (OKP) hierarchies and universal character  hierarchy of B-type (BUC hierarchy),  which are  proved as  zero modes of certain combinations of the generating functions. By applying the strategy of carrying out the action of the quantum fields on vacuum vector, the generating functions for symplectic Schur function, orthogonal Schur function and generalized $Q$-function have been presented. The remarkable feature is that polynomial tau-functions are the coefficients of certain family of generating functions. Furthermore, in terms of  the Vandermonde-like identity and properties of Pfaffian, it is showed that the polynomial tau-functions of the SKP, OKP and BUC hierarchies can be written as determinant and Pfaffian forms, respectively. In addition, the soliton solutions of the SKP and OKP hierarchies have been discussed.

\textbf{Keywords}: polynomial tau-functions; symplectic/orthogonal Schur function;  generalized $Q$-function; generating functions; soliton solutions\\
\textbf{ Mathematics Subject Classifications (2000)}: 17B80, 35Q55, 37K10
\end{abstract}
\maketitle
\tableofcontents
\section{Introduction}

Symmetric functions are the characters of the irreducible highest weight representation of the classical groups \cite{Weyl1946}-\cite{Macdonald1995}, which play a significant role in mathematical physics especially in the theory of integrable systems \cite{Jing1991}-\cite{BEsymmetricgroup2013}.  Schur functions and Schur $Q$-functions are the basic symmetric functions which are the solutions of differential equations in  Kadomtsev-Petviashvili (KP) and Kadomtsev-Petviashvili sub-hierarchy of B-type (BKP) hierarchies, respectively. In the famous work of the Kyoto School, the authors \cite{Sato1981}-\cite{J5} investigated the core connection of the infinite dimensional Lie algebra and their highest weight vectors to the integrable hierarchies  involving KP, BKP,  discrete KP (DKP), modified KP (MKP) and $s$-component KP hierarchies.

Koike \cite{Koike1989} introduced a polynomial with a pair of partitions called the universal character, which is  a generalization of Schur function. The universal character (UC) hierarchy,  proposed by Tsuda \cite{Tsuda2004}, as an infinite-dimensional integrable system satisfied by  the universal character. It can be regarded as a extension of the KP hierarchy.  Furthermore,  in the subsequent paper \cite{Tsuda2005}-\cite{Tsuda2012}, Tsuda presented the relations between the $q$-Painlev\'{e} equations and  the lattice $q$-UC hierarchy which are the extended $q$-KP  and $q$-UC hierarchies. Based upon this,  the  structure and properties of  the Painlev\'{e} equations and their higher order analogues have been developed, such as rational solutions, Lax formalism, bilinear relations for $\tau$-functions and Weyl group symmetry.  Wang et.al \cite{Wang2019} discussed the algebra of universal characters and the phase model of strongly correlated bosons. Recently, by means of the vertex operator realization of symplectic and orthogonal Schur functions \cite{Baker1996}, the authors \cite{Wang2020,Wang2021} established  generalizations of symplectic KP (SKP) and orthogonal KP (OKP) hierarchies called the symplectic  and orthogonal universal character hierarchies corresponding to symplectic and orthogonal universal characters, respectively. Ogawa \cite{Ogawa2009} defined a generalized $Q$-function expressed by the Pfaffian and constructed an integrable UC hierarchy of B-type (BUC hierarchy) characterized by the generalized $Q$-function. Lately, in the paper \cite{Li2019,Li2021}, the author generalized the theory of BUC hierarchy to a coupled and plethystic cases, which can derive coupled  and plethystic  infinite order nonlinear PDEs.

All the polynomial tau-functions of the KP hierarchy can be expressed as a disjoint union of Schubert cells and the Schur polynomial is the center of the Schubert cell. You \cite{You1989, You1991} showed polynomial tau-functions of the BKP, DKP and modified DKP  (MDKP) hierarchies  all include the $Q$-Schur polynomials which are the centers of Schubert cells of the infinite-dimensional orthogonal Grassmann manifold. Kac et al. \cite{Kac2018} constructed all the  polynomial tau-functions of the KP and MKP hierarchies from Schur polynomials by some shift of arguments. Moreover, the polynomial tau-functions of the BKP, DKP, and MDKP hierarchies have been well discussed in terms of boson-fermion correspondence \cite{Kac2019}. Then by means of the $s$-comonent boson-fermion correspondence, Kac et al. \cite{Kac2019M} have studied the polynomial tau-functions of the multi-component KP hierarchy.  Based on the quantum fields, they also develop the twisted quantum fields presentation of  Hall-Littlewood polynomials and derived a novel deformed boson-fermion correspondence \cite{Vecochea2020}. Rozhkovskaya \cite{Rozhkovskaya2019}  proved multiparameter Schur $Q$-functions are tau-functions of the BKP hierarchy. Besides, in the frame work of  quantum fields presentation and generating functions of the symmetric functions, recent study has shown that the polynomial tau-functions of the KP, BKP and $s$-component KP hierarchies can be expressed as the zero-modes of certain combinations of generating functions \cite{Kac2021}. Recently, the polynomial tau-functions of the UC and  multi-component UC hierarchies have also been analyzed \cite{LiCZ2022}. To our best knowledge, there is no existing references on the study on the polynomial tau-functions of the SKP, OKP and BUC hierarchies. Based upon the facts, in the view point of the quantum field presentation of the symmetric functions, we will concentrate on the construction of
the exact solutions of these integrable hierarchies including polynomial-type and soliton-type solutions. We shall prove the polynomial tau-functions of SKP, OKP and BUC hierarchies can be regarded as  zero modes of certain combinatorial generating functions.

The present paper is organized as follows. In Section 2, we begin with a review of the elementary, complete symmetric functions, power sums and Schur polynomial. Section 3 is devoted to construction of quantum fields  of symplectic Schur functions and the polynomial tau-functions of the SKP hierarchy. Meanwhile, the $n$-soliton solutions of this integrable system are derived. In section 4, the generating functions for the orthogonal Schur functions are investigated by the action of the operators on the vacuum vector $1$.
In terms of quantum fields presentation, we also study the polynomial tau-functions and $n$-soliton solutions of the SKP hierarchy. The fact that the  polynomial tau-functions of the BUC hierarchy are the coefficients of certain family generating functions is described in Section 5. The last Section are conclusions and discussions.

\section{Preliminaries on symmetric functions}
In this section, we mainly retrospect some basic facts and properties about symmetric functions.

Let $\Lambda(\mathbf{x})$ be the ring of symmetric functions in variables $\mathbf{x}=(x_{1},x_{2}\ldots)$. The $r$th elementary symmetric function $e_{r}$, complete symmetric function $h_{r}$ and  power sum $p_{r}$ are defined by (cf. \cite{Macdonald1995})
\begin{eqnarray}\
&&e_{r}(\mathbf{x})=\sum_{i_{1}<i_{2}<\ldots<i_{r}<\infty}x_{i_{1}}x_{i_{2}}\ldots x_{i_{r}}, \quad  for \quad r\geq1, \notag\\
&&h_{r}(\mathbf{x})=\sum_{i_{1}\leq i_{2}\leq\ldots\leq i_{r}\leq \infty}x_{i_{1}}x_{i_{2}}\ldots x_{i_{r}}, \quad  for \quad r\geq1, \notag\\
&&p_{r}(\mathbf{x})=\sum_{i}x_{i}^{r}.
\end{eqnarray}
It is universally known that $h_{r}(\mathbf{x})=e_{r}(\mathbf{x})=p_{r}(\mathbf{x})=0$ for $r<0$ and $h_{0}=e_{0}=p_{0}=1$. The generating functions for these symmetric functions are
\begin{eqnarray} \label{gf1}
E(u)&=&\sum_{k\geq0}e_{k}(\mathbf{x})u^{k}=\prod_{i\geq1}\left(1+x_{i}u\right)=\exp\left(-\sum_{n\geq1}\frac{(-1)^{n}p_{n}}{n}u^{n}\right), \notag\\
H(u)&=&\sum_{k\geq0}h_{k}(\mathbf{x})u^{k}=\prod_{i\geq1}\frac{1}{1-x_{i}u}=\exp\left(\sum_{n\geq1}\frac{p_{n}}{n}u^{n}\right), \notag\\
P(u)&=&\sum_{k\geq1}p_{k}(\mathbf{x})u^{k-1}=\frac{H^{\prime}(u)}{H(u)}.
\end{eqnarray}
It is easy to obtain that $H(u)E(-u)=1$.

The polynomials $S_{k}(t_{1},t_{2},\ldots)$, $k\in \mathbb{Z}$, is determined by the generating function
\begin{equation}\ \label{Schur Taylor}
\sum_{k\in \mathbb{Z}}S_{k}(t_{1},t_{2},\ldots)u^{k}=\exp\left(\sum_{j=1}^{\infty}t_{j}u^{j}\right).
\end{equation}
The Schur polynomial $S_{\lambda}(t_{1},t_{2},\ldots)$ can be expressed as (cf. \cite{Macdonald1995})
\begin{equation}\
S_{\lambda}(t_{1},t_{2},\ldots)=\det[S_{\lambda_{i}-i+j}(t_{1},t_{2},\ldots)]_{1\leq i,j\leq l},
\end{equation}
where  $\lambda=(\lambda_{1},\ldots,\lambda_{l})$ is a partition.

For each symmetric function $f\in\Lambda$, let $f^{\bot}:\Lambda\rightarrow\Lambda$ be the adjoint of multiplication by $f$
\begin{equation}\
\langle f^{\bot}g,\omega\rangle=\langle g,f\omega\rangle,\quad g,f,\omega\in\Lambda.
\end{equation}
Let us consider generating functions of the adjoint operators
\begin{eqnarray}\
&&E^{\bot}(u)=\sum_{k\geq0}\frac{e^{\bot}_{k}}{u^{k}},\quad  H^{\bot}(u)=\sum_{k\geq0}\frac{h^{\bot}_{k}}{u^{k}}.
\end{eqnarray}
It is straightforward to show that
\begin{eqnarray}\ \label{jiont}
&&E^{\bot}(u)=\exp\left(-\sum_{k\geq1}(-1)^{k}\frac{\partial}{\partial{p_{k}}}\frac{1}{u^{k}}\right),\quad  H^{\bot}(u)=\exp\left(\sum_{k\geq1}\frac{\partial}{\partial{p_{k}}}\frac{1}{u^{k}}\right).
\end{eqnarray}
\begin{proposition}\label{generating relation}
The generating functions satisfy the following relations  (cf. \cite{Macdonald1995})
\begin{eqnarray}
&&\left(1-\frac{v}{u}\right)E^{\bot}(u)E(v)=E(v)E^{\bot}(u),\notag\\
&&\left(1-\frac{v}{u}\right)H^{\bot}(u)H(v)=H(v)H^{\bot}(u),\notag\\
&&H^{\bot}(u)E(v)=\left(1+\frac{v}{u}\right)E(v)H^{\bot}(u),\notag\\
&&E^{\bot}(u)H(v)=\left(1+\frac{v}{u}\right)H(v)E^{\bot}(u).
\end{eqnarray}
\end{proposition}

\section{Polynomial tau-functions and $n$-soliton solutions of the SKP  hierarchy}

By means of charged free fermions, we devote to discussing structures and properties of polynomial tau-functions for the SKP hierarchy. Furthermore, the generating functions for polynomial tau-functions of the SKP hierarchy can be obtained by acting the quantum fields  of symplectic Schur functions on the bosonic Fock space $\mathcal{B}^{m}$. There is an interesting conclusion that the polynomial tau-functions of the SKP hierarchy are the coefficients of certain family of generating functions. Finally, the soliton-type solutions of the SKP hierarchy have been derived.

\subsection{Quantum fields presentation of symplectic Schur functions and the SKP hierarchy}

The symmetric polynomial ring  $\Lambda$ : $\Lambda=\mathbb{C}[e_1,e_2,\ldots]=\mathbb{C}[h_1,h_2,\ldots]=\mathbb{C}[p_1,p_2,\ldots]$ can be generated by elementary, complete symmetric functions and power sums, respectively. Introduce the bosonic Fock space $\mathcal{B}=\mathbb{C}[z,z^{-1}]\bigotimes \Lambda $, it is decomposed to obtain the charged graded space
\begin{eqnarray}
\mathcal{B}=\bigoplus\limits_{m\in\mathbb{Z}}\mathcal{B}^m, \ \ \text{where} \ \ \mathcal{B}^{m}=z^{m}\cdot\mathbb{C}[p_{1},p_{2},\ldots]=z^{m}\Lambda.
\end{eqnarray}

Let $R(u)$ act on the elements of the form $z^{m}f,\ f\in \Lambda,\ m\in \mathbb{Z}$,
$R(u):\mathcal{B}\rightarrow\mathcal{B} $  is defined as (cf. \cite{Kac2013})
\begin{eqnarray}\ \label{RD}
&&R(u)\left(z^{m}f\mathbf{(x)}\right)=z^{m+1}u^{m+1}f.
\end{eqnarray}
Then it leads to
\begin{eqnarray}
&&R^{-1}(u)\left(z^{m}f\mathbf{(x)}\right)=z^{m-1}u^{-m}f.
\end{eqnarray}
Operators $R^{\pm1}(u)$ map the grading of the boson Fock space $\mathcal{B}^{(m)}$ into $\mathcal{B}^{(m\pm1)}$.

Define the quantum fields $\psi^{Sp,\pm}(u)$ \cite{Li2022}
\begin{eqnarray}\ \label{symplectic Schur vertex operators}
&&\psi^{Sp,+}(u)=u^{-1}R(u)H(u)E^{\bot}(-u)E^{\bot}(-\frac{1}{u})=\sum\limits_{k\in\mathbb{Z}+\frac{1}{2}}\psi^{Sp,+}_{k}
u^{-k-\frac{1}{2}}, \notag\\
&&\psi^{Sp,-}(u)=(1-u^{2})R^{-1}(u)E(-u)H^{\bot}(u)H^{\bot}(\frac{1}{u})=\sum\limits_{k\in\mathbb{Z}+\frac{1}{2}}\psi^{Sp,-}_{k}
u^{-k-\frac{1}{2}} .
\end{eqnarray}

\begin{proposition}
 Quantum fields $\psi^{Sp,+}(u),\psi^{Sp,-}(u)$ satisfy the anticommutation relations
\begin{eqnarray}\ \label{symplectic Schur relations}
&&\psi^{Sp,\pm}(u)\psi^{Sp,\pm}(v)+\psi^{Sp,\pm}(v)\psi^{Sp,\pm}(u)=0, \notag\\
&&\psi^{Sp,+}(u)\psi^{Sp,-}(v)+\psi^{Sp,-}(v)\psi^{Sp,+}(u)=\delta(u,v),
\end{eqnarray}
where $\delta(u,v)=\sum\limits_{{k,m\in\mathbb{Z}}\atop{k+m=-1}}u^{k}v^{m}$ is the delta-distribution.

From Eq.(\ref{symplectic Schur vertex operators}), Eq.(\ref{symplectic Schur relations}) is equivalent to the relations with charged free fermions
\begin{eqnarray}\ \label{symplectic fenliang}
&&\psi^{Sp,\pm}_{k}\psi^{Sp,\pm}_{l}+\psi^{Sp,\pm}_{l}\psi^{Sp,\pm}_{k}=0,\notag\\
&&\psi^{Sp,+}_{k}\psi^{Sp,-}_{l}+\psi^{Sp,-}_{l}\psi^{Sp,+}_{k}=\delta_{k,-l}.
\end{eqnarray}
\end{proposition}

\begin{remark} \label{SKP remark}
From Eqs. (\ref{gf1}) and (\ref{jiont}), we easily get the bosonic form of the quantum fields $\psi^{Sp,\pm}(u)$:
\begin{eqnarray}\
&&\psi^{Sp,+}(u)=u^{-1}R(u)\exp\left(\sum\limits_{n\geq1}\frac{p_{n}}{n}u^{n}\right)\exp\left(-\sum\limits_{n\geq1}\frac{\partial}{\partial p_{n}}\frac{1}{u^{n}}\right)\exp\left(-\sum\limits_{n\geq1}\frac{\partial}{\partial p_{n}}u^{n}\right), \notag\\
&&\psi^{Sp,-}(u)=(1-u^{2})R^{-1}(u)\exp\left(-\sum\limits_{n\geq1}\frac{p_{n}}{n}u^{n}\right)\exp\left(\sum\limits_{n\geq1}\frac{\partial}{\partial p_{n}}\frac{1}{u^{n}}\right)\exp\left(\sum\limits_{n\geq1}\frac{\partial}{\partial p_{n}}u^{n}\right).
\end{eqnarray}
Hence one has $\psi^{Sp,+}_{i}(z^{m})=0$ if $i>-m-\frac{1}{2}$ and $\psi^{Sp,-}_{i}(z^{m})=0$ if $i>m-\frac{1}{2}$.
\end{remark}

\begin{definition}
For an unknown function $\tau=\tau(x)$, the bilinear equation
\begin{eqnarray}\ \label{SKP bilinear equation}
&&\widehat{\Omega}(\tau\otimes\tau)=0,
\end{eqnarray}
is called the SKP hierarchy, where
\begin{eqnarray}\
&&\widehat{\Omega}=\sum\limits_{l\in\mathbb{Z}+\frac{1}{2}}\psi_{l}^{Sp,+}\otimes\psi_{-l}^{Sp,-}.
\end{eqnarray}
\end{definition}

\begin{lemma} \label{SKP 1}
Let $\widehat{X}=\sum\limits_{i>N}C_{i}\psi_{i}^{Sp,+}$, where $C_{i}\in\mathbb{C}, N\in\mathbb{Z}$. Then $\widehat{X}^{2}=0$.
\end{lemma}

\begin{proof}
Due to $\psi^{Sp,+}_{k}\psi^{Sp,+}_{l}+\psi^{Sp,+}_{l}\psi^{Sp,+}_{k}=0$, $i,k\in\mathbb{Z}+\frac{1}{2}$, one immediately has
\begin{eqnarray}\
\widehat{X}^{2}=\widehat{X}\cdot \widehat{X}=\sum\limits_{l>N}C_{l}\psi_{l}^{Sp,+}\cdot \sum\limits_{k>N}C_{k}\psi_{k}^{Sp,+}=\sum\limits_{l>N}\sum\limits_{k>N}C_{l}C_{k}
\psi_{l}^{Sp,+}\psi_{k}^{Sp,+}=0.
\end{eqnarray}
\end{proof}

\begin{lemma} \label{SKP 2}
Let $\widehat{X}=\sum\limits_{i>N}C_{i}\psi_{i}^{Sp,+}$, where $C_{i}\in\mathbb{C}, N\in\mathbb{Z}$. Then $\widehat{\Omega}(\widehat{X}\otimes \widehat{X})=(\widehat{X}\otimes \widehat{X})\widehat{\Omega}$.
\end{lemma}

\begin{proof}
Based on $\psi_{-l}^{Sp,-}\widehat{X}=-\widehat{X}\psi_{-l}^{Sp,-}+C_{l}$, we have
\begin{eqnarray}\
\widehat{\Omega}(\widehat{X}\otimes\widehat{X})&=&\sum\limits_{l\in\mathbb{Z}+\frac{1}{2}}\psi_{l}^{Sp,+}\widehat{X}\otimes\psi_{-l}^{Sp,-}\widehat{X}
=\sum\limits_{l\in\mathbb{Z}+\frac{1}{2}}(-\widehat{X}\psi_{l}^{Sp,+})\otimes(-\widehat{X}\psi_{-l}^{Sp,-}+C_{l})\notag\\
&=&(\widehat{X}\otimes\widehat{X})\widehat{\Omega}-\widehat{X}\sum\limits_{l\in\mathbb{Z}+\frac{1}{2}}C_{l}\psi_{l}^{Sp,+}\otimes1
=(\widehat{X}\otimes\widehat{X})\widehat{\Omega}-\widehat{X}^{2}\otimes1=(\widehat{X}\otimes\widehat{X})\widehat{\Omega}.
\end{eqnarray}
\end{proof}

\begin{corollary} \label{SKP corollary 2}
Let $\tau \in \mathcal{B}^m$ be a tau-function of the SKP hierarchy, and let $\widehat{X}=\sum\limits_{i>N}C_{i}\psi_{i}^{Sp,+}$, where $C_{i}\in\mathbb{C}, N\in\mathbb{Z}$. Then $\widehat{\tau}=\widehat{X}\tau \in \mathcal{B}^{m+1}$ is also a tau-functions of the SKP hierarchy.
\end{corollary}

\begin{proof}
Multiplying $\widehat{X}\otimes\widehat{X}$  left on both sides of $\widehat{\Omega}(\tau\otimes\tau)=0$, we get $(\widehat{X}\otimes\widehat{X})\widehat{\Omega}(\tau\otimes\tau)=0$. According to $\widehat{\Omega}(\widehat{X}\otimes \widehat{X})=(\widehat{X}\otimes \widehat{X})\widehat{\Omega}$, it follows that $(\widehat{X}\otimes\widehat{X})\widehat{\Omega}(\tau\otimes\tau)=\widehat{\Omega}(\widehat{X}\otimes\widehat{X})(\tau\otimes\tau)=\widehat{\Omega}(\widehat{X}\tau
\otimes\widehat{X}\tau)=0$. Therefore, $\widehat{X}\tau$ is the solution of the SKP hierarchy.
\end{proof}

It is known that symplectic Schur functions are tau-functions of the SKP hierarchy \cite{Wang2020}. If the non-zero solution of (\ref{SKP bilinear equation}) is a polynomial function of variables $(p_{1}, p_{2},\ldots)$, we call the non-zero solution a polynomial tau-function. It follows from  Remark \ref{SKP remark} that $z^{m}$ is a solution of the SKP hierarchy. We now turn our attention to the polynomial tau-function of the SKP hierarchy.

\subsection{Generating functions and the polynomial tau-functions of the SKP hierarchy}

Let $\widehat{G}(u_{1},\ldots,u_{l})$ be a generating function of the symplectic Schur function in $\mathbf{u}=(u_{1},\ldots,u_{l})$ defined by
\begin{eqnarray}\ \label{SKP generating function}
\widehat{G}(u_{1},\ldots,u_{l})=\prod\limits_{1\leq i<j\leq l}\left(u_{i}-u_{j}\right)\left(1-u_{i}u_{j}\right)\prod\limits_{i=1}^{l}H(u_{i}).
\end{eqnarray}

From Proposition \ref{generating relation}, we have
\begin{eqnarray}
&&\psi^{Sp,+}(u_1)\psi^{Sp,+}(u_2)\cdots\psi^{Sp,+}(u_l)(z^{k}f)\notag\\
&=&z^{k+l}u_{1}^{l+k-1}\cdots u_{l-1}^{k+1}u_{l}^{k}\prod\limits_{1\leq i<j\leq l}\left(1-u_{i}u_{j}\right)\left(1-\frac{u_j}{u_i}\right)\prod\limits_{i=1}^{l}H(u_{i})
E^{\bot}(-u_i)E^{\bot}(-\frac{1}{u_i})(f)\notag\\
&=&z^{k+l}u_{1}^{k}\cdots u_{l}^{k}\widehat{G}(u_{1},\ldots,u_{l}).
\end{eqnarray}

Let $\widehat{A}_{1}(u),\ldots,\widehat{A}_{l}(u)$ be the set of formal Laurent series, define the formal Laurent series $\widehat{T}_{i}(u)=\widehat{A}_{i}(u)H(u)=\sum\limits_{p\in\mathbb{Z}}\widehat{T}_{i,p}u^{p},\ i=1,\ldots,l$. Besides, let $\widehat{T}(u_{1},\ldots,u_{l})$ be a formal Laurent series in $(u_{1},\ldots,u_{l})$ defined by
\begin{eqnarray}\
\widehat{T}(u_{1},\ldots,u_{l})=\prod\limits_{1\leq i<j\leq l}\left(u_{i}-u_{j}\right)\left(1-u_{i}u_{j}\right)\prod\limits_{i=1}^{l}\widehat{A}_{i}(u_{i})H(u_{i}).
\end{eqnarray}
For any vector $\boldsymbol\xi=(\xi_{1},\ldots,\xi_{l})\in \mathbb{Z}^{l}$, $\widehat{T}_{\boldsymbol\xi}$ is the coefficient of the following expansion
\begin{eqnarray}\ \label{SKP T}
\widehat{T}(u_{1},\ldots,u_{l})=\sum\limits_{\boldsymbol\xi\in\mathbb{Z}^{l}}\widehat{T}_{\boldsymbol\xi}u_{1}^{\xi_{1}}\cdots u_{l}^{\xi_{l}}.
\end{eqnarray}

\begin{theorem} \label{SKP theorem}
\begin{itemize}
\item [1)]Formal Laurent series $\widehat{T}(u_{1},\ldots,u_{l})$ can be expressed as
\begin{eqnarray}\
&&\widehat{T}(u_{1},\ldots,u_{l})=\frac{1}{2}\det\left[\left(u_{i}^{l-j}+u_{i}^{l+j-2}\right)\widehat{T}_{i}(u_{i})\right]_{1\leq i,j\leq l}.
\end{eqnarray}
\item [2)]
The coefficient $\widehat{T}_{\boldsymbol\xi}$ of $u_{1}^{\xi_{1}}\cdots u_{l}^{\xi_{l}}$ in (\ref{SKP T}) can be written as follows
\begin{eqnarray}\
&&\widehat{T}_{\boldsymbol\xi}=\frac{1}{2}\det\left[\widehat{T}_{i,\xi_{i}-j}+\widehat{T}_{i,\xi_{i}+j-2}\right]_{1\leq i,j\leq l},
\end{eqnarray}
where $\boldsymbol\xi=(\xi_{1},\ldots,\xi_{l})$.
\item [3)]
$\widehat{T}_{\boldsymbol\xi}$ is a polynomial tau-function of the SKP hierarchy.
\end{itemize}
\end{theorem}

\begin{proof}
\begin{itemize}
  \item [1)]According to Vandermonde-like identity \cite{Jing2015}
  \begin{eqnarray}
  \det\left[u_{i}^{k-j}+u_{i}^{k+j-2}\right]&=&2\prod\limits_{1\leq i<j\leq k}\left(u_{i}-u_{j}\right)\left(1-u_{i}u_{j}\right)\notag\\
  &=&\sum\limits_{{\sigma\in S_{k}}\atop{\varepsilon_{i}=\pm1}}sgn(\sigma)(u_{1}\cdots u_{k})^{k-1}u_{1}^{\varepsilon_{1}(\sigma(1)-1)}\cdots u_{k}^{\varepsilon_{k}(\sigma(k)-1)},
  \end{eqnarray}
it is easy to verify that
  \begin{eqnarray}
  \widehat{T}(u_{1},\ldots,u_{l})&=&\prod\limits_{1\leq i<j\leq l}\left(u_{i}-u_{j}\right)\left(1-u_{i}u_{j}\right)\prod\limits_{i=1}^{l}\widehat{A}_{i}(u_{i})H(u_{i})\notag\\
  &=&\frac{1}{2}\det\left[u_{i}^{l-j}+
  u_{i}^{l+j-2}\right]\prod\limits_{i=1}^{l}\widehat{T}_{i}(u_{i})\notag\\
  &=&\frac{1}{2}\det\left[\left(u_{i}^{l-j}+u_{i}^{l+j-2}\right)\widehat{T}_{i}(u_{i})\right]_{1\leq i,j\leq l}.
  \end{eqnarray}
  \item [2)]Observe that
  \begin{align}
  \widehat{T}(u_{1},\ldots,u_{l})&=\frac{1}{2}\det\left[\sum\limits_{p_{i}\in\mathbb{Z}}\widehat{T}_{i,p_{i}}\left(u_{i}^{l+p_{i}-j}+u_{i}^{l+p_{i}+j-2}\right)\right]\notag\\
  &=\frac{1}{2}\sum\limits_{p_{i}\in\mathbb{Z}}\sum\limits_{{\sigma\in S_{l}}\atop{\varepsilon_{i}=\pm1}}sgn(\sigma)u_{1}^{l+p_{1}-1}\cdots u_{l}^{l+p_{l}-1}\widehat{T}_{1,p_{1}}u_{1}^{\varepsilon_{1}(\sigma(1)-1)}\cdots \widehat{T}_{l,p_{l}}u_{l}^{\varepsilon_{l}(\sigma(l)-1)}\notag\\
  &=\sum\limits_{\xi_{i}\in\mathbb{Z}}\frac{1}{2}\sum\limits_{{\sigma\in S_{l}}\atop{\varepsilon_{i}=\pm1}}sgn(\sigma)\widehat{T}_{1,\xi_{1}-l+1-\varepsilon_{1}(\sigma(1)-1)}\cdots \widehat{T}_{l,\xi_{l}-l+1-\varepsilon_{l}(\sigma(l)-1)}u_{1}^{\xi_{1}}\cdots u_{l}^{\xi_{l}}\notag\\
  &=\sum\limits_{\xi_{i}\in\mathbb{Z}}\frac{1}{2}\det\left[\widehat{T}_{i,\xi_{i}-j}+\widehat{T}_{i,\xi_{i}+j-2}\right]_{1\leq i,j\leq l}u_{1}^{\xi_{1}}\cdots u_{l}^{\xi_{l}},
  \end{align}
  therefore, the coefficient $\widehat{T}_{\boldsymbol\xi}$ of $u_{1}^{\xi_{1}}\cdots u_{l}^{\xi_{l}}$ is $\frac{1}{2}\det\left[\widehat{T}_{i,\xi_{i}-j}+\widehat{T}_{i,\xi_{i}+j-2}\right]_{1\leq i,j\leq l}$.
  \item [3)]It is apparent from (\ref{SKP generating function}) that
  \begin{eqnarray}
  &&\widehat{A}_{1}(u_{1})\cdots \widehat{A}_{l}(u_{l})\psi^{Sp,+}(u_{1})\cdots \psi^{Sp,+}(u_{l})(z^{k}\cdot1)=z^{l+k}u_{1}^{k}\cdots u_{l}^{k}\widehat{T}(u_{1},\ldots,u_{l}).
  \end{eqnarray}
  Let $\widehat{A}_{j}(u)=\sum\limits_{M_{j}\leq r\leq N_{j}}
  \widehat{A}_{j,r-\frac{1}{2}}u^{r}(A_{j,r-\frac{1}{2}}\in\mathbb{C}, M_{j},N_{j},r\in\mathbb{Z}, j=1,\ldots,l)$ be a power series expansion of the variable $u$. Therefore, $\widehat{T}_{\boldsymbol\xi}$ can be written as
  \begin{eqnarray}
  &&\widehat{T}_{\boldsymbol\xi}=z^{-l-k}\widehat{X}_{1}\cdots \widehat{X}_{l}(z^{k}\cdot1),
  \end{eqnarray}
  where
  \begin{eqnarray}
  &&\widehat{X}_{j}=\sum\limits_{M_{j}-\xi_{j}-k-\frac{1}{2}\leq i_{j}\leq N_{j}-\xi_{j}-k-\frac{1}{2}}\widehat{A}_{j, \xi_{j}+k+i_{j}}\psi^{Sp,+}_{i_{j}},\quad j=1,\cdots,l.
  \end{eqnarray}
  Particularly, by Remark \ref{SKP remark} and Corollary \ref{SKP corollary 2}, the coefficient $\widehat{T}_{\boldsymbol\xi}$ is a tau-function of the SKP hierarchy with $k=0$. Since $\widehat{T}_{\boldsymbol\xi}$ is a finite linear combination of $\psi_{i_{1}}^{Sp,+}\cdots\psi_{i_{l}}^{Sp,+}(1)$, it is a polynomial tau-function.
  \end{itemize}
  \end{proof}

By replacing $\widehat{A}_{j}(u)$ with $u^{\xi_{j}}\widehat{A}_{j}(u)$,  we have
   \begin{eqnarray}\
  &&\widehat{A}_{i}(u)=u^{N_{i}}\widehat{h}_{i}\sum\limits_{k=0}^{\infty}a_{i,k}u^{k},\quad \quad N_{i}\in\mathbb{Z}, \widehat{h}_{i},a_{i,k}\in\mathbb{C}, a_{i,0}=1, \widehat{h}_{i}\neq0, i=1,\ldots,l,
  \end{eqnarray}
where  $\widehat{A}_{j}(u)$ are non-zero Laurent series defined in the $\widehat{T}(u_{1},\ldots,u_{l})$.

According to (\ref{Schur Taylor}), $\sum\limits_{k=0}^{\infty}a_{i,k}u^{k}$ can be written as
   \begin{eqnarray}\
   &&\sum\limits_{k=0}^{\infty}a_{i,k}u^{k}=\exp\left(\sum\limits_{l=1}^{\infty}\widehat{c}_{i,l}u^{l}\right),\quad \text{and} \quad a_{i,k}=S_{k}(\widehat{c}_{i,1},\widehat{c}_{i,2},\ldots),
  \end{eqnarray}
  where $\{\widehat{c}_{i,l}\}$ is a  set of constants in $\mathbb{C}$.

  From (\ref{gf1}) and  setting $t_{l}=\frac{p_{l}}{l}$, we obtain
  \begin{eqnarray}
  \widehat{T}_{i}(u)&=&\widehat{A}_{i}(u)H(u)=u^{N_{i}}\widehat{h}_{i}\exp\left(\sum\limits_{l=1}^{\infty}\widehat{c}_{i,l}u^{l}\right)\exp\left(\sum\limits_{l=1}^{\infty}
  \frac{p_{l}}{l}u^{l}\right)\notag\\
  &=&u^{N_{i}}\widehat{h}_{i}\exp\left(\sum\limits_{l=1}^{\infty}(\widehat{c}_{i,l}+t_{l})u^{l}\right)=u^{N_{i}}\widehat{h}_{i}\sum\limits_{l=0}^{\infty}
  S_{l}(t_{1}+\widehat{c}_{i,1},t_{2}+\widehat{c}_{i,2},\ldots)u^{l}.
  \end{eqnarray}
  Hence $\widehat{T}_{i,p}=\widehat{h}_{i}S_{p-N_{i}}(t_{1}+\widehat{c}_{i,1},t_{2}+\widehat{c}_{i,2},\ldots),\ i=1,\ldots,l $.  From Theorem \ref{SKP theorem}, polynomial tau-functions of the SKP hierarchy are given by
  \begin{eqnarray} \label{pau} \widehat{T}_{\boldsymbol\xi}&=&\frac{1}{2}\det\left[\widehat{T}_{i,\xi_{i}-j}+\widehat{T}_{i,\xi_{i}+j-2}\right]\notag\\
  &=&\frac{1}{2}\det\left[\widehat{h}_{i}S_{\xi_{i}-j-N_{i}}(t_{1}+\widehat{c}_{i,1},t_{2}+\widehat{c}_{i,2},\ldots)+\widehat{h}_{i}S_{\xi_{i}+j-2-N_{i}}
  (t_{1}+\widehat{c}_{i,1},t_{2}+\widehat{c}_{i,2},\ldots)\right]\notag\\
  &=&\prod\limits_{i=1}^{l}\widehat{h}_{i}\frac{1}{2}\det\left[S_{\xi_{i}-j-N_{i}}(t_{1}+\widehat{c}_{i,1},t_{2}+\widehat{c}_{i,2},\ldots)+S_{\xi_{i}+j-2-N_{i}}
  (t_{1}+\widehat{c}_{i,1},t_{2}+\widehat{c}_{i,2},\ldots)\right]_{i,j=1,\ldots,l}.\notag\\
  \end{eqnarray}
  When $\widehat{c}_{i,l}=0$, $\widehat{h}_{i}=1$ and $N_{i}+2=i$, $\widehat{T}_{\boldsymbol\xi}$ reduces to the symplectic Schur functions \cite{Wang2020}. The polynomial tau-functions (\ref{pau}) of the SKP hierarchy are the generalization of the solution of the SKP hierarchy in \cite{Wang2020}, which are the zero mode of an appropriate combinatorial generating functions.

\subsection{$N$-soliton solutions of the SKP hierarchy}
Now let us consider  another extremely important exact solution of SKP hierarchy called the soliton solution.

Let
\begin{eqnarray}\
&&\Gamma^{Sp}(p,q)=p^{-1}(1-q^{2})R(p)R^{-1}(q)H(p)E(-q)E^{\bot}(-p)H^{\bot}(q)E^{\bot}(-\frac{1}{p})H^{\bot}(\frac{1}{q}).
\end{eqnarray}
From Proposition \ref{generating relation}, it is easy to check that
\begin{eqnarray}\ \label{SKP soliton relation}
&&\Gamma^{Sp}(p_{i},q_{i})\Gamma^{Sp}(p_{j},q_{j})=A_{ij}:\Gamma^{Sp}(p_{i},q_{i})\Gamma^{Sp}(p_{j},q_{j}):,
\end{eqnarray}
where
\begin{eqnarray}\
&&A_{ij}=\frac{(1-p_{i}p_{j})(1-q_{i}q_{j})(p_{i}-p_{j})(q_{i}-q_{j})}{(1-q_{i}p_{j})(1-p_{i}q_{j})(q_{i}-p_{j})(p_{i}-q_{j})}.
\end{eqnarray}
In particular $\Gamma^{Sp}(p,q)^{2}=0$; therefore $e^{pc\Gamma^{Sp}(p,q)}=1+pc\Gamma^{Sp}(p,q)$.

\begin{lemma} \label{SKP tau lemma1}
If $\tau$ is a solution of the SKP hierarchy, then $\Gamma^{Sp}(u,v)\tau$ is also a solution.
\end{lemma}
\begin{proof}
Using Eq. (\ref{symplectic fenliang}), we obtain
\begin{eqnarray}
&&\widehat{\Omega}\left(\psi^{Sp,+}(u)\psi^{Sp,-}(v)\otimes\psi^{Sp,+}(u)\psi^{Sp,-}(v)\right)\notag\\
&=&(\sum\limits_{l\in\mathbb{Z}+\frac{1}{2}}\psi_{l}^{Sp,+}\otimes\psi_{-l}^{Sp,-})(\sum\limits_{m,n\in\mathbb{Z}+\frac{1}{2}}\psi_{m}^{Sp,+}u^{-m-\frac{1}{2}}
\psi_{n}^{Sp,-}v^{-n-\frac{1}{2}}\otimes\sum\limits_{m,n\in\mathbb{Z}+\frac{1}{2}}\psi_{m}^{Sp,+}u^{-m-\frac{1}{2}}\psi_{n}^{Sp,-}v^{-n-\frac{1}{2}})\notag\\
&=&\sum\limits_{l\in\mathbb{Z}+\frac{1}{2}}(\sum\limits_{m,n\in\mathbb{Z}+\frac{1}{2}}\psi_{l}^{Sp,+}\psi_{m}^{Sp,+}\psi_{n}^{Sp,-}u^{-m-\frac{1}{2}}v^{-n-\frac{1}{2}}
\otimes\sum\limits_{m,n\in\mathbb{Z}+\frac{1}{2}}\psi_{-l}^{Sp,-}\psi_{m}^{Sp,+}\psi_{n}^{Sp,-}u^{-m-\frac{1}{2}}v^{-n-\frac{1}{2}})\notag\\
&=&\sum\limits_{l,m,n\in\mathbb{Z}+\frac{1}{2}}-\psi_{m}^{Sp,+}\delta_{l,-n}u^{-m-\frac{1}{2}}v^{-n-\frac{1}{2}}\otimes\delta_{m,l}\psi_{n}^{Sp,-}
u^{-m-\frac{1}{2}}v^{-n-\frac{1}{2}}+\psi_{m}^{Sp,+}\psi_{n}^{Sp,-}\psi_{l}^{Sp,+}u^{-m-\frac{1}{2}}v^{-n-\frac{1}{2}}\notag\\
&&\otimes\psi_{m}^{Sp,+}\psi_{n}^{Sp,-}\psi_{-l}^{Sp,-}u^{-m-\frac{1}{2}}v^{-n-\frac{1}{2}}-\psi_{m}^{Sp,+}\delta_{l,-n}u^{-m-\frac{1}{2}}v^{-n-\frac{1}{2}}
\otimes\psi_{m}^{Sp,+}\psi_{n}^{Sp,-}\psi_{-l}^{Sp,-}u^{-m-\frac{1}{2}}v^{-n-\frac{1}{2}}\notag\\
&&+\psi_{m}^{Sp,+}\psi_{n}^{Sp,-}\psi_{l}^{Sp,+}u^{-m-\frac{1}{2}}v^{-n-\frac{1}{2}}\otimes\delta_{m,l}\psi_{n}^{Sp,-}u^{-m-\frac{1}{2}}v^{-n-\frac{1}{2}}\notag\\
&=&\sum\limits_{l,m,n\in\mathbb{Z}+\frac{1}{2}}\psi_{m}^{Sp,+}\psi_{n}^{Sp,-}\psi_{l}^{Sp,+}u^{-m-\frac{1}{2}}v^{-n-\frac{1}{2}}\otimes
\psi_{m}^{Sp,+}\psi_{n}^{Sp,-}\psi_{-l}^{Sp,-}u^{-m-\frac{1}{2}}v^{-n-\frac{1}{2}}\notag\\
&=&\left(\psi^{Sp,+}(u)\psi^{Sp,-}(v)\otimes\psi^{Sp,+}(u)\psi^{Sp,-}(v)\right)\widehat{\Omega}.
\end{eqnarray}
It can easily be checked that
\begin{eqnarray}
&&\left(\psi^{Sp,+}(u)\psi^{Sp,-}(v)\otimes\psi^{Sp,+}(u)\psi^{Sp,-}(v)\right)\widehat{\Omega}(\tau\otimes\tau)\notag\\
&=&\widehat{\Omega}\left(\psi^{Sp,+}(u)\psi^{Sp,-}(v)\otimes\psi^{Sp,+}(u)\psi^{Sp,-}(v)\right)(\tau\otimes\tau)\notag\\
&=&\widehat{\Omega}\left(\psi^{Sp,+}(u)\psi^{Sp,-}(v)\tau\otimes\psi^{Sp,+}(u)\psi^{Sp,-}(v)\tau\right)=0.
\end{eqnarray}
Clearly, $\psi^{Sp,+}(u)\psi^{Sp,-}(v)\tau$ is the solution of the SKP hierarchy. A routine computation gives rise to  $\psi^{Sp,+}(u)\psi^{Sp,-}(v)=\frac{1}{1-uv}\frac{u}{u-v}\Gamma^{Sp}(u,v)$. Therefore, $\Gamma^{Sp}(u,v)\tau$ is also a solution.
\end{proof}

\begin{lemma} \label{SKP tau lemma2}
It holds that
\begin{eqnarray}
&&[\widehat{\Omega},1\otimes\Gamma^{Sp}(p,q)+\Gamma^{Sp}(p,q)\otimes1]=0,
\end{eqnarray}
where $[A,B]=_{def}AB-BA$.
\end{lemma}

\begin{proof}
Lemma can be calculated directly from the $\Gamma^{Sp}(p,q)=(1-pq)\frac{p-q}{p}\psi^{Sp,+}(p)\psi^{Sp,-}(q)$. The specific calculation process is not listed here.
\end{proof}

Let us consider the  function
\begin{eqnarray}\ \label{SKP n solition}
\tau(x,y)=\tau(x,y;p,q,c)=\prod\limits_{i=1}^{n}e^{p_{i}c_{i}\Gamma^{Sp}(p_{i},q_{i})}\cdot1,
\quad p_{i},q_{i},c_{i}\in \mathbb{C},p_{i}\neq q_{j},p_{i}\neq \frac{1}{q_{j}} for \ i\neq j,
\end{eqnarray}
 and set
\begin{eqnarray}\
\eta_{i}=\sum\limits_{k\geq 1}(p_{i}^{k}-q_{i}^{k})\frac{p_{k}(x)}{k}.
\end{eqnarray}
By (\ref{SKP soliton relation}), Eq.(\ref{SKP n solition}) can be rewritten as
\begin{eqnarray}\ \label{SKP 而}
&&\tau(x,y;p,q,c)=\sum_{J \subset I}\left(\prod_{i\in J}c_{i}(1-q_{i}^{2})\right)\left(\prod\limits_{{i,j\in J}\atop{i<j}}A_{ij}\right)\exp\left(\sum\limits_{i\in J}\eta_{i}\right),
\end{eqnarray}
where $I=\{1,2,\ldots,n\}$ .

\begin{proposition} \label{SKP soliton tau}
The function $\tau(x,y;p,q,c)$ in (\ref{SKP 而}) is a solution of the SKP hierarchy, which we call the $n$-soliton solutions.
\end{proposition}

\begin{proof}
Suppose that $\tau$ is a solution of the SKP hierarchy. We put $\widehat{\tau}=\left(1+pc\Gamma^{Sp}(p,q)\right)\tau$. It follows from Lemma \ref{SKP tau lemma1} and \ref{SKP tau lemma2} that
\begin{eqnarray}
\widehat{\Omega}(\widehat{\tau}\otimes\widehat{\tau})&=&\widehat{\Omega}(\tau\otimes\tau)+pc\widehat{\Omega}(\tau\otimes\Gamma^{Sp}(p,q)\tau+\Gamma^{Sp}(p,q)\tau\otimes\tau)
+p^{2}c^{2}\widehat{\Omega}(\Gamma^{Sp}(p,q)\tau\otimes\Gamma^{Sp}(p,q)\tau)\notag\\
&=&pc\widehat{\Omega}[(1\otimes\Gamma^{Sp}(p,q))(\tau\otimes\tau)+(\Gamma^{Sp}(p,q)\otimes1)(\tau\otimes\tau)]\notag\\
&=&pc\widehat{\Omega}(1\otimes\Gamma^{Sp}(p,q)+\Gamma^{Sp}(p,q)\otimes1)(\tau\otimes\tau)\notag\\
&=&pc(1\otimes\Gamma^{Sp}(p,q)+\Gamma^{Sp}(p,q)\otimes1)\widehat{\Omega}(\tau\otimes\tau)\notag\\
&=&0.
\end{eqnarray}
Hence $\widehat{\tau}$ is a solution of the SKP hierarchy. Note that $\tau=1$ solves the SKP hierarchy, it is easy to see that the $n$-soliton solutions defined in (\ref{SKP n solition}) is really a solution of the SKP hierarchy.

\end{proof}

\section{Polynomial tau-functions and $n$-soliton solutions of the OKP  hierarchy}

In this section, we firstly construct quantum fields of orthogonal Schur functions and deduce the relationship between these operators. Meanwhile, the generating functions of the orthogonal Schur functions have been investigated. Moreover, by applying the quantum field presentation of the OKP hierarchy, the polynomial tau-functions and the soliton solutions have been presented.

\subsection{Quantum fields presentation of orthogonal Schur functions and the OKP hierarchy}
Introduce the quantum fields defined by
\begin{eqnarray}\ \label{orthogonal Schur vertex operators}
&&\psi^{O,+}(u)=u^{-1}(1-u^{2})R(u)H(u)E^{\bot}(-u)E^{\bot}(-\frac{1}{u})=\sum\limits_{k\in\mathbb{Z}+\frac{1}{2}}\psi^{O,+}_{k}
u^{-k-\frac{1}{2}}, \notag\\
&&\psi^{O,-}(u)=R^{-1}(u)E(-u)H^{\bot}(u)H^{\bot}(\frac{1}{u})=\sum\limits_{k\in\mathbb{Z}+\frac{1}{2}}\psi^{O,-}_{k}
u^{-k-\frac{1}{2}} .
\end{eqnarray}

\begin{proposition}
 It can be checked that $\psi^{O,+}(u),\psi^{O,-}(u)$ satisfy the  relations
\begin{eqnarray}\ \label{orthogonal Schur relations}
&&\psi^{O,\pm}(u)\psi^{O,\pm}(v)+\psi^{O,\pm}(v)\psi^{O,\pm}(u)=0, \notag\\
&&\psi^{O,+}(u)\psi^{O,-}(v)+\psi^{O,-}(v)\psi^{O,+}(u)=\delta(u,v) .
\end{eqnarray}
Equivalently, Eq.(\ref{orthogonal Schur relations}) can be expressed as charged free fermions relation
\begin{eqnarray}\
&&\psi^{O,\pm}_{k}\psi^{O,\pm}_{l}+\psi^{O,\pm}_{l}\psi^{O,\pm}_{k}=0,\notag\\
&&\psi^{O,+}_{k}\psi^{O,-}_{l}+\psi^{O,-}_{l}\psi^{O,+}_{k}=\delta_{k,-l}.
\end{eqnarray}
\end{proposition}

\begin{remark} \label{OKP remark}
From the formula (\ref{gf1}) and (\ref{jiont}), we easily get the bosonic form of the fields $\psi^{O,\pm}(u)$:
\begin{eqnarray}\
&&\psi^{O,+}(u)=u^{-1}(1-u^{2})R(u)\exp\left(\sum\limits_{n\geq1}\frac{p_{n}}{n}u^{n}\right)\exp\left(-\sum\limits_{n\geq1}\frac{\partial}{\partial p_{n}}\frac{1}{u^{n}}\right)\exp\left(-\sum\limits_{n\geq1}\frac{\partial}{\partial p_{n}}u^{n}\right), \notag\\
&&\psi^{O,-}(u)=R^{-1}(u)\exp\left(-\sum\limits_{n\geq1}\frac{p_{n}}{n}u^{n}\right)\exp\left(\sum\limits_{n\geq1}\frac{\partial}{\partial p_{n}}\frac{1}{u^{n}}\right)\exp\left(\sum\limits_{n\geq1}\frac{\partial}{\partial p_{n}}u^{n}\right).
\end{eqnarray}
Hence we obtain $\psi^{O,+}_{i}(z^{m})=0$ if $i>-m-\frac{1}{2}$ and $\psi^{O,-}_{i}(z^{m})=0$ if $i>m-\frac{1}{2}$.
\end{remark}

\begin{definition}
For an unknown function $\tau=\tau(x)$, the bilinear equation
\begin{eqnarray}\ \label{OKP bilinear equation}
&&\widetilde{\Omega}(\tau\otimes\tau)=0,
\end{eqnarray}
is called the OKP hierarchy, where
\begin{eqnarray}\
&&\widetilde{\Omega}=\sum\limits_{k\in\mathbb{Z}+\frac{1}{2}}\psi_{k}^{O,+}\otimes\psi_{-k}^{O,-}.
\end{eqnarray}
\end{definition}

\begin{lemma} \label{OKP 1}
Let $\widetilde{X}=\sum\limits_{i>N}C_{i}\psi_{i}^{O,+}$, where $C_{i}\in\mathbb{C}, N\in\mathbb{Z}$. Then $\widetilde{X}^{2}=0$.
\end{lemma}

\begin{lemma} \label{OKP 2}
Let $\widetilde{X}=\sum\limits_{i>N}C_{i}\psi_{i}^{O,+}$, where $C_{i}\in\mathbb{C}, N\in\mathbb{Z}$. Then $\widetilde{\Omega}(\widetilde{X}\otimes \widetilde{X})=(\widetilde{X}\otimes \widetilde{X})\widetilde{\Omega}$.
\end{lemma}

\begin{corollary} \label{OKP corollary 2}
Let $\tau \in \mathcal{B}^m$ be a tau-function of the OKP hierarchy, and let $\widetilde{X}=\sum\limits_{i>N}C_{i}\psi_{i}^{O,+}$, where $C_{i}\in\mathbb{C}, N\in\mathbb{Z}$. Then $\widetilde{\tau}=\widetilde{X}\tau \in \mathcal{B}^{m+1}$ is also a tau-functions of the OKP hierarchy.
\end{corollary}

\begin{proof}
The proof of the Lemma \ref{OKP 1}, \ref{OKP 2} and Corollary \ref{OKP corollary 2} is quite similar to the Lemma \ref{SKP 1}, \ref{SKP 2} and Corollary \ref{SKP corollary 2}, so is omitted.
\end{proof}

\subsection{Generating functions and the polynomial tau-functions of the OKP hierarchy}

It is  known that orthogonal Schur functions are tau-functions of the OKP hierarchy \cite{Wang2021}. Let $\widetilde{G}(u_{1},\ldots,u_{l})$ be a generating function of the orthogonal Schur function in $\mathbf{u}=(u_{1},\ldots,u_{l})$ defined by
\begin{eqnarray}\ \label{OKP generating function}
\widetilde{G}(u_{1},\ldots,u_{l})=\prod\limits_{1\leq i<j\leq l}\left(u_{i}-u_{j}\right)\prod\limits_{1\leq i\leq j\leq l}\left(1-u_{i}u_{j}\right)\prod\limits_{i=1}^{l}H(u_{i}).
\end{eqnarray}

By Proposition \ref{generating relation}, we obatin
\begin{eqnarray}
&&\psi^{O,+}(u_1)\cdots\psi^{O,+}(u_l)(z^{k}f)=z^{k+l}u_{1}^{k}\cdots u_{l}^{k}\widetilde{G}(u_{1},\ldots,u_{l}).
\end{eqnarray}

Consider the set of formal Laurent series $\widetilde{A}_{1}(u),\ldots,\widetilde{A}_{l}(u)$, define the formal Laurent series $\widetilde{T}_{i}(u)=\widetilde{A}_{i}(u)H(u)=\sum\limits_{p\in\mathbb{Z}}\widetilde{T}_{i,p}u^{p},\ i=1,\ldots,l$. Besides, let $\widetilde{T}(u_{1},\ldots,u_{l})$ be a formal Laurent series in $(u_{1},\ldots,u_{l})$ defined by
\begin{eqnarray}\
\widetilde{T}(u_{1},\ldots,u_{l})=\prod\limits_{1\leq i<j\leq l}\left(u_{i}-u_{j}\right)\prod\limits_{1\leq i\leq j\leq l}\left(1-u_{i}u_{j}\right)\prod\limits_{i=1}^{l}\widetilde{A}_{i}(u_{i})H(u_{i}).
\end{eqnarray}
For any vector $\boldsymbol\zeta=(\zeta_{1},\ldots,\zeta_{l})\in \mathbb{Z}^{l}$, $\widetilde{T}_{\boldsymbol\zeta}$ is the coefficient of the following expansion
\begin{eqnarray}\ \label{OKP T}
\widetilde{T}(u_{1},\ldots,u_{l})=\sum\limits_{\boldsymbol\zeta\in\mathbb{Z}^{l}}\widetilde{T}_{\boldsymbol\zeta}u_{1}^{\zeta_{1}}\cdots u_{l}^{\zeta_{l}}.
\end{eqnarray}

\begin{theorem} \label{OKP theorem}
\begin{itemize}
\item [1)] Formal Laurent series $\widetilde{T}(u_{1},\ldots,u_{l})$ can be written as
\begin{eqnarray}\
&&\widetilde{T}(u_{1},\ldots,u_{l})=\det\left[\left(u_{i}^{l-j}-u_{i}^{l+j}\right)\widetilde{T}_{i}(u_{i})\right]_{1\leq i,j\leq l}.
\end{eqnarray}
\item [2)]
For any vector $\boldsymbol\zeta=(\zeta_{1},\ldots,\zeta_{l})\in \mathbb{Z}^{l}$,  the coefficient $\widetilde{T}_{\boldsymbol\zeta}$ of $u_{1}^{\zeta_{1}}\cdots u_{l}^{\zeta_{l}}$ in (\ref{OKP T}) is given by
\begin{eqnarray}\
&&\widetilde{T}_{\boldsymbol\zeta}=\det\left[\widetilde{T}_{i,\zeta_{i}-l-j}-\widetilde{T}_{i,\zeta_{i}-l+j}\right]_{1\leq i,j\leq l}.
\end{eqnarray}
\item [3)]
$\widetilde{T}_{\boldsymbol\zeta}$ is a polynomial tau-function of the OKP hierarchy.
\end{itemize}
\end{theorem}

\begin{proof}
\begin{itemize}
  \item [1)]According to Vandermonde-like identity \cite{Jing2015}
  \begin{eqnarray}
  \det\left[u_{i}^{k-j}-u_{i}^{k+j}\right]&=&\prod\limits_{1\leq i<j\leq k}\left(u_{i}-u_{j}\right)\prod\limits_{1\leq i\leq j\leq k}\left(1-u_{i}u_{j}\right)\notag\\
  &=&\sum\limits_{{\sigma\in S_{k}}\atop{\varepsilon_{i}=\pm1}}sgn(\sigma)\varepsilon_{1}\cdots\varepsilon_{k}u_{1}^{k-\varepsilon_{1}\sigma(1)}\cdots u_{k}^{k-\varepsilon_{k}\sigma(k)},
  \end{eqnarray}
  we have
  \begin{eqnarray}
  \widetilde{T}(u_{1},\ldots,u_{l})&=&\prod\limits_{1\leq i<j\leq l}\left(u_{i}-u_{j}\right)\prod\limits_{1\leq i\leq j\leq l}\left(1-u_{i}u_{j}\right)\prod\limits_{i=1}^{l}\widetilde{A}_{i}(u_{i})H(u_{i})\notag\\
  &=&\det\left[u_{i}^{l-j}-
  u_{i}^{l+j}\right]\prod\limits_{i=1}^{l}\widetilde{T}_{i}(u_{i})=\det\left[\left(u_{i}^{l-j}-u_{i}^{l+j}\right)\widetilde{T}_{i}(u_{i})\right]_{1\leq i,j\leq l}.
  \end{eqnarray}
  \item [2)]Noticing that
  \begin{align}
  \widetilde{T}(u_{1},\ldots,u_{l})&=\det\left[\sum\limits_{p_{i}\in\mathbb{Z}}\widetilde{T}_{i,p_{i}}\left(u_{i}^{l+p_{i}-j}-u_{i}^{l+p_{i}+j}\right)\right]\notag\\
  &=\sum\limits_{p_{i}\in\mathbb{Z}}\sum\limits_{{\sigma\in S_{l}}\atop{\varepsilon_{i}=\pm1}}sgn(\sigma)\varepsilon_{1}\cdots\varepsilon_{l}\widetilde{T}_{1,p_{1}}u_{1}^{l+p_{1}-\varepsilon_{1}\sigma(1)}
  \cdots\widetilde{T}_{l,p_{l}}u_{l}^{l+p_{l}-\varepsilon_{l}\sigma(l)}\notag\\
  &=\sum\limits_{\zeta_{i}\in\mathbb{Z}}\sum\limits_{{\sigma\in S_{l}}\atop{\varepsilon_{i}=\pm1}}sgn(\sigma)\varepsilon_{1}\cdots\varepsilon_{l}\widetilde{T}_{1,\zeta_{1}-l+\varepsilon_{1}\sigma(1)}\cdots \widetilde{T}_{l,\zeta_{l}-l+\varepsilon_{l}\sigma(l)}u_{1}^{\zeta_{1}}\cdots u_{l}^{\zeta_{l}}\notag\\
  &=\sum\limits_{\zeta_{i}\in\mathbb{Z}}\det\left[\widetilde{T}_{i,\zeta_{i}-l-j}-\widetilde{T}_{i,\zeta_{i}-l+j}\right]_{1\leq i,j\leq l}u_{1}^{\zeta_{1}}\cdots u_{l}^{\zeta_{l}}.
  \end{align}
  Obviously, the coefficient $\widetilde{T}_{\boldsymbol\zeta}$ of $u_{1}^{\zeta_{1}}\cdots u_{l}^{\zeta_{l}}$ is $\det\left[\widetilde{T}_{i,\zeta_{i}-l-j}-\widetilde{T}_{i,\zeta_{i}-l+j}\right]_{1\leq i,j\leq l}$.
  \item [3)]From (\ref{OKP generating function}), it is straightforward to show that
  \begin{eqnarray}
  &&\widetilde{A}_{1}(u_{1})\cdots \widetilde{A}_{l}(u_{l})\psi^{O,+}(u_{1})\cdots \psi^{O,+}(u_{l})(z^{k}\cdot1)=z^{l+k}u_{1}^{k}\cdots u_{l}^{k}\widetilde{T}(u_{1},\ldots,u_{l}).
  \end{eqnarray}
  Let $\widetilde{A}_{j}(u)=\sum\limits_{M_{j}\leq r\leq N_{j}}\widetilde{A}_{j,r-\frac{1}{2}}u^{r}(\widetilde{A}_{j,r-\frac{1}{2}}\in\mathbb{C}, M_{j},N_{j},r\in\mathbb{Z}, j=1,\ldots,l)$ be a power series expansion of the variable $u$. Therefore,  $\widetilde{T}_{\boldsymbol\zeta}$  can be written as
  \begin{eqnarray}
  &&\widetilde{T}_{\boldsymbol\zeta}=z^{-l-k}\widetilde{X}_{1}\cdots \widetilde{X}_{l}(z^{k}\cdot1),
  \end{eqnarray}
  where
  \begin{eqnarray} \label{OKP dengjia 1}
  &&\widetilde{X}_{j}=\sum\limits_{M_{j}-\zeta_{j}-k-\frac{1}{2}\leq i_{j}\leq N_{j}-\zeta_{j}-k-\frac{1}{2}}\widetilde{A}_{j,\zeta_{j}+k+i_{j}}\psi^{O,+}_{i_{j}},\quad j=1,\cdots,l.
  \end{eqnarray}

 By Remark \ref{OKP remark} and Corollary \ref{OKP corollary 2}, it should be pointed out that  the coefficient $\widetilde{T}_{\boldsymbol\zeta}$ is a tau-function of the OKP hierarchy with $k=0$. Since $\widetilde{T}_{\boldsymbol\zeta}$ is a finite linear combination of $\psi_{i_{1}}^{O,+}\cdots\psi_{i_{l}}^{O,+}(1)$, it is a polynomial tau-function.
\end{itemize}
\end{proof}

By changing $\widetilde{A}_{j}(u)\rightarrow u^{\zeta_{j}}\widetilde{A}_{j}(u)$, we obtain
   \begin{eqnarray}\
  &&\widetilde{A}_{i}(u)=u^{N_{i}}\widetilde{h}_{i}\sum\limits_{k=0}^{\infty}\widetilde{a}_{i,k}u^{k},\ \  N_{i}\in\mathbb{Z}, \widetilde{h}_{i},\widetilde{a}_{i,k}\in\mathbb{C}, \widetilde{a}_{i,0}=1, \widetilde{h}_{i}\neq0, i=1,\ldots,l.
  \end{eqnarray}
From (\ref{Schur Taylor}), $\sum\limits_{k=0}^{\infty}\widetilde{a}_{i,k}u^{k}$ can be expressed as
\begin{eqnarray}
&&\sum\limits_{k=0}^{\infty}\widetilde{a}_{i,k}u^{k}=\exp\left(\sum\limits_{l=1}^{\infty}\widetilde{c}_{i,l}u^{l}\right),\ \ \text{and} \ \  \widetilde{a}_{i,k}=S_{k}(\widetilde{c}_{i,1},\widetilde{c}_{i,2},\ldots),
\end{eqnarray}
where $\{\widetilde{c}_{i,l}\}$ are constants in $\mathbb{C}$. Then based on (\ref{gf1}), we get
  \begin{eqnarray}\
  \widetilde{T}_{i}(u)&=&\widetilde{A}_{i}(u)H(u)=u^{N_{i}}\widetilde{h}_{i}\sum\limits_{l=0}^{\infty}
  S_{l}(t_{1}+\widetilde{c}_{i,1},t_{2}+\widetilde{c}_{i,2},\ldots)u^{l}.
  \end{eqnarray}
  Hence $\widetilde{T}_{i,p}=\widetilde{h}_{i}S_{p-N_{i}}(t_{1}+\widetilde{c}_{i,1},t_{2}+\widetilde{c}_{i,2},\ldots),\ i=1,\ldots,l $.  From Theorem \ref{OKP theorem}, polynomial tau-functions of the OKP hierarchy have the form
\begin{eqnarray}\label{ptau}
\widetilde{T}_{\boldsymbol\zeta}&=&\det\left[\widetilde{T}_{i,\zeta_{i}-l-j}-\widetilde{T}_{i,\zeta_{i}-l+j}\right]\notag\\
&=&\det\left[\widetilde{h}_{i}S_{\zeta_{i}-l-j-N_{i}}(t_{1}+\widetilde{c}_{i,1},t_{2}+\widetilde{c}_{i,2},\ldots)-\widetilde{h}_{i}S_{\zeta_{i}-l+j-N_{i}}
(t_{1}+\widetilde{c}_{i,1},t_{2}+\widetilde{c}_{i,2},\ldots)\right]\notag\\
&=&\prod\limits_{i=1}^{l}(\widetilde{h}_{i})\det\left[S_{\zeta_{i}-l-j-N_{i}}(t_{1}+\widetilde{c}_{i,1},t_{2}+\widetilde{c}_{i,2},\ldots)-S_{\zeta_{i}-l+j-N_{i}}
(t_{1}+\widetilde{c}_{i,1},t_{2}+\widetilde{c}_{i,2},\ldots)\right]_{i,j=1,\ldots,l}.\notag\\
\end{eqnarray}
Under the reduction $\widetilde{c}_{i,l}=0$, $\widetilde{h}_{i}=-1$ and $l+N_{i}=i$, $\widetilde{T}_{\boldsymbol\zeta}$ lead to the orthogonal Schur functions \cite{Wang2021}. Thus the polynomial tau-functions (\ref{ptau}) of the OKP hierarchy can be reduced to the solution of the OKP hierarchy in \cite{Wang2021}, which are the zero mode of an appropriate combinatorial generating functions.

\subsection{$N$-soliton solutions of the OKP hierarchy}
Let
\begin{eqnarray}\
&&\Gamma^{O}(p,q)=p^{-1}(1-p^{2})R(p)R^{-1}(q)H(p)E(-q)E^{\bot}(-p)H^{\bot}(q)E^{\bot}(-\frac{1}{p})H^{\bot}(\frac{1}{q}).
\end{eqnarray}
From Proposition \ref{generating relation}, it is easy to check that
\begin{eqnarray}\ \label{OKP soliton relation}
&&\Gamma^{O}(p_{i},q_{i})\Gamma^{O}(p_{j},q_{j})=A_{ij}:\Gamma^{O}(p_{i},q_{i})\Gamma^{O}(p_{j},q_{j}):,
\end{eqnarray}
where
\begin{eqnarray}\
&&A_{ij}=\frac{(1-p_{i}p_{j})(1-q_{i}q_{j})(p_{i}-p_{j})(q_{i}-q_{j})}{(1-q_{i}p_{j})(1-p_{i}q_{j})(q_{i}-p_{j})(p_{i}-q_{j})}.
\end{eqnarray}
In particular $\Gamma^{O}(p,q)^{2}=0$; therefore $e^{pc\Gamma^{O}(p,q)}=1+pc\Gamma^{O}(p,q)$.

\begin{lemma} \label{OKP tau 1}
If $\tau$ is a solution of the OKP hierarchy, then $\Gamma^{O}(u,v)\tau$ is also a solution.
\end{lemma}

\begin{lemma} \label{OKP tau 2}
It holds that
\begin{eqnarray}
&&[\widetilde{\Omega},1\otimes\Gamma^{O}(p,q)+\Gamma^{O}(p,q)\otimes1]=0.
\end{eqnarray}
\end{lemma}

Considering the following function
\begin{eqnarray}\ \label{OKP n solition}
\tau(x,y)=\tau(x,y;p,q,c)=\prod\limits_{i=1}^{n}e^{p_{i}c_{i}\Gamma^{O}(p_{i},q_{i})}\cdot1,
\quad p_{i},q_{i},c_{i}\in \mathbb{C},p_{i}\neq q_{j},p_{i}\neq \frac{1}{q_{j}} for \ i\neq j.
\end{eqnarray}
Let us set
\begin{eqnarray}\
\eta_{i}=\sum\limits_{k\geq 1}(p_{i}^{k}-q_{i}^{k})\frac{p_{k}(x)}{k}.
\end{eqnarray}
From Eq.(\ref{OKP soliton relation}), Eq.(\ref{OKP n solition}) can be rewritten as
\begin{eqnarray}\ \label{OKP 而}
&&\tau(x,y;p,q,c)=\sum_{J \subset I}\left(\prod_{i\in J}c_{i}(1-p_{i}^{2})\right)\left(\prod\limits_{{i,j\in J}\atop{i<j}}A_{ij}\right)\exp\left(\sum\limits_{i\in J}\eta_{i}\right),
\end{eqnarray}
where $I=\{1,2,\ldots,n\}$ .

\begin{proposition} \label{OKP soliton tau}
The function $\tau(x,y;p,q,c)$ in (\ref{OKP 而}) is a solution of the OKP hierarchy, which we call the $n$-soliton solutions.
\end{proposition}

\begin{proof}
Lemma \ref{OKP tau 1}, \ref{OKP tau 2} and Proposition \ref{OKP soliton tau} can be proved with the similar procedure as in Lemma \ref{SKP tau lemma1}, \ref{SKP tau lemma2} and Proposition \ref{SKP soliton tau}.
\end{proof}

\section{Polynomial tau-functions of the BUC hierarchy}

In this section, the quantum fields  of the generalized $Q$-functions shall be developed. By using neutral fermions, we construct an integrable BUC hierarchy characterized by the generalized $Q$-functions. Based upon the generating functions of the polynomial tau-functions of the BUC hierarchy, it is showed that the polynomial tau-function of the BUC hierarchy is a zero mode of certain generating functions.

\subsection{Quantum fields presentation of the generalized $Q$-functions and the BUC hierarchy}
Introduce another class of symmetric functions $q_{k}(x_{1},x_{2},\ldots)$ by
\begin{eqnarray}\
&&Q(u)=\sum\limits_{k\in \mathbb{Z}}q_{k}u^{k}=E(u)H(u),
\end{eqnarray}
here $q_{k}=\sum\limits_{i=0}^{k}e_{i}h_{k-i}$ for $k>0$, $q_{0}=1$ and $q_{k}=0$ for $k<0$.

Define
\begin{eqnarray}\ \label{BUC Q S}
&&Q(u)=S(u)^{2},\quad \text{where} \ S(u)=\exp\left(\sum\limits_{n\in N_{odd}}\frac{p_{n}}{n}u^{n}\right),\notag\\
&&S^{\bot}(u)=\exp\left(\sum\limits_{n\in N_{odd}}\frac{\partial}{\partial p_{n}}\frac{1}{u^{n}}\right),\quad N_{odd}=\{1,3,5,\ldots\}.
\end{eqnarray}

\begin{proposition} \label{BUC generating jiaohuan}
The following commutation relations about generating functions hold (cf.\cite{Rozhkovskaya2019})
\begin{eqnarray}\
&&H^{\bot}(u)Q(v)=\frac{u+v}{u-v}Q(v)H^{\bot}(u),\notag\\
&&E^{\bot}(u)Q(v)=\frac{u+v}{u-v}Q(v)E^{\bot}(u),\notag\\
&&S^{\bot}(u)Q(v)=\frac{u+v}{u-v}Q(v)S^{\bot}(u).
\end{eqnarray}
\end{proposition}

Define the formal distributions $\varphi(u)$ and $\overline{\varphi}(u)$ of operators acting on  the boson Fock space $\mathcal{B}_{odd}=\mathbb{C}[p_{1},p_{3},p_{5},\ldots]$
\begin{eqnarray}\
&&\varphi(u)=Q(u)S^{\prime\bot}(-\frac{1}{u})S^{\bot}(-u)=\sum\limits_{j\in\mathbb{Z}}\varphi_{j}u^{-j},\notag\\
&&\overline{\varphi}(u)=Q^{\prime}(u)S^{\bot}(-\frac{1}{u})S^{\prime\bot}(-u)=\sum\limits_{j\in\mathbb{Z}}\overline{\varphi}_{j}u^{-j},
\end{eqnarray}
where $Q^\prime(u)$ means the generating functions for the $q_{k}(\mathbf{y})$ and their adjoint operators hold for the variable $\mathbf{y}$, the operators $\varphi_{j}$ and $\overline{\varphi}_{j}$ are the neutral fermions.

Let
\begin{eqnarray}\
&&f(u,v)=\frac{u-v}{u+v}=1+2\sum\limits_{k\geq1}(-1)^{k}\frac{v^{k}}{u^{k}},\quad |u|>|v|,
\end{eqnarray}
then
\begin{eqnarray}\ \label{BUC 汛}
&&f(u,v)+f(v,u)=(v-u)\delta(v,-u)=2\sum\limits_{k\in\mathbb{Z}}\frac{v^{k}}{(-u)^{k}}=2v\delta(v,-u).
\end{eqnarray}

\begin{proposition}
$\varphi(u)$ and $\overline{\varphi}(u)$ satisfy the following relations
\begin{eqnarray}\ \label{BUC jiaohuan}
&&\varphi(u)\varphi(v)+\varphi(v)\varphi(u)=2v\delta(v,-u),\notag\\
&&\overline{\varphi}(u)\overline{\varphi}(v)+\overline{\varphi}(v)\overline{\varphi}(u)=2v\delta(v,-u),\notag\\
&&\varphi(u)\overline{\varphi}(v)-\overline{\varphi}(v)\varphi(u)=0.
\end{eqnarray}
Eq.(\ref{BUC jiaohuan}) can alos be expressed as neutral fermions relation
\begin{eqnarray}\ \label{BUC fenliang}
&&\varphi_{m}\varphi_{n}+\varphi_{n}\varphi_{m}=2(-1)^{m}\delta_{m+n,0},\notag\\
&&\overline{\varphi}_{m}\overline{\varphi}_{n}+\overline{\varphi}_{n}\overline{\varphi}_{m}=2(-1)^{m}\delta_{m+n,0},\notag\\
&&\varphi_{m}\overline{\varphi}_{n}-\overline{\varphi}_{n}\varphi_{m}=0.
\end{eqnarray}
\end{proposition}

\begin{proof}
We only prove the first formula of (\ref{BUC jiaohuan}) and (\ref{BUC fenliang}),  other formulas cab be proved similarly. In terms of Proposition \ref{BUC generating jiaohuan} and (\ref{BUC 汛}), we have
\begin{eqnarray}\
\varphi(u)\varphi(v)&=&Q(u)S^{\prime\bot}(-\frac{1}{u})S^{\bot}(-u)Q(v)S^{\prime\bot}(-\frac{1}{v})S^{\bot}(-v)\notag\\
&=&\frac{-u+v}{-u-v}Q(u)Q(v)S^{\prime\bot}(-\frac{1}{u})S^{\bot}(-u)S^{\prime\bot}(-\frac{1}{v})S^{\bot}(-v),\\
\varphi(v)\varphi(u)&=&Q(v)S^{\prime\bot}(-\frac{1}{v})S^{\bot}(-v)Q(u)S^{\prime\bot}(-\frac{1}{u})S^{\bot}(-u)\notag\\
&=&\frac{-v+u}{-v-u}Q(v)Q(u)S^{\prime\bot}(-\frac{1}{v})S^{\bot}(-v)S^{\prime\bot}(-\frac{1}{u})S^{\bot}(-u),
\end{eqnarray}
therefore,
\begin{eqnarray}\
\varphi(u)\varphi(v)+\varphi(v)\varphi(u)&=&\left(\frac{u-v}{u+v}+\frac{v-u}{v+u}\right)Q(u)Q(v)S^{\prime\bot}(-\frac{1}{u})
S^{\bot}(-u)S^{\prime\bot}(-\frac{1}{v})S^{\bot}(-v)\notag\\
&=&(v-u)\delta(v,-u)Q(u)Q(v)S^{\prime\bot}(-\frac{1}{u})S^{\bot}(-u)S^{\prime\bot}(-\frac{1}{v})S^{\bot}(-v)\notag\\
&=&2v\delta(v,-u).
\end{eqnarray}
Expanding  $\varphi(u)\varphi(v)+\varphi(v)\varphi(u)=2v\delta(v,-u)$ into
\begin{eqnarray}
&&\sum\limits_{m\in\mathbb{Z}}\varphi_{m}u^{-m}\sum\limits_{n\in\mathbb{Z}}\varphi_{n}v^{-n}+
\sum\limits_{n\in\mathbb{Z}}\varphi_{n}v^{-n}\sum\limits_{m\in\mathbb{Z}}\varphi_{m}u^{-m}=
2v\sum\limits_{k\in\mathbb{Z}}\frac{v^{k}}{(-u)^{k+1}},
\end{eqnarray}
and taking the coefficient of $u^{-m}v^{-n}$ at both ends of the above formula, we derive the first formula of (\ref{BUC fenliang}).
\end{proof}

\begin{remark} \label{BUC 1 solution}
From the formula (\ref{BUC Q S}), we derive the bosonic form of the quantum fields $\varphi(u)$ and $\overline{\varphi}(u)$ as follows
\begin{eqnarray}
\varphi(u)&=&\exp\left(\sum\limits_{n\in\mathbb{N}_{odd}}\frac{2p_{n}}{n}u^{n}\right)\exp\left(-\sum\limits_{n\in\mathbb{N}_{odd}}
\frac{\partial}{\partial p_{n}^{\prime}}u^{n}\right)\exp\left(-\sum\limits_{n\in\mathbb{N}_{odd}}\frac{\partial}{\partial p_{n}}\frac{1}{u^{n}}\right),\notag\\
\overline{\varphi}(u)&=&\exp\left(\sum\limits_{n\in\mathbb{N}_{odd}}\frac{2p_{n}^{\prime}}{n}u^{n}\right)\exp\left(-\sum\limits_{n\in\mathbb{N}_{odd}}
\frac{\partial}{\partial p_{n}}u^{n}\right)\exp\left(-\sum\limits_{n\in\mathbb{N}_{odd}}\frac{\partial}{\partial p_{n}^{\prime}}\frac{1}{u^{n}}\right).
\end{eqnarray}
Hence one can check that $\varphi_{m}(1)=0(m>0)$, $\varphi_{0}(1)=1$ and $\overline{\varphi}_{n}(1)=0(n>0)$, $\overline{\varphi}_{0}(1)=1$.
\end{remark}

\begin{definition}
The BUC hierarchy is the system of bilinear relations
\begin{eqnarray} \label{BUC bilinear}
&&\Omega(\tau\otimes\tau)=\overline{\Omega}(\tau\otimes\tau)=(\tau\otimes\tau),
\end{eqnarray}
where
\begin{eqnarray}
&&\Omega=\sum\limits_{n\in\mathbb{Z}}\varphi_{n}\otimes(-1)^{n}\varphi_{-n},\ \ \ \
\overline{\Omega}=\sum\limits_{n\in\mathbb{Z}}\overline{\varphi}_{n}\otimes(-1)^{n}\overline{\varphi}_{-n}.
\end{eqnarray}
\end{definition}

From Remark \ref{BUC 1 solution}, it is easy to see that $\tau=1$ is a tau-function of the BUC hierarchy. Similarly, if the solution of (\ref{BUC bilinear}) is  a polynomial function of the variables $(p_{1}, p_{3}, \ldots)$, we say it is a polynomial tau-function. Now we consider other forms of tau-functions of the BUC hierarchy.

\begin{lemma}
Let $X=\sum\limits_{n\geq N}A_{n}\varphi_{n}$, $Y=\sum\limits_{m\geq M}B_{m}\overline{\varphi}_{m}$, where $A_{n},B_{m}\in\mathbb{C}$ and $N,M\in\mathbb{Z}$. Then
\begin{eqnarray} \label{BUC X Y}
&&X^{2}=\left\{
\begin{aligned}
&\sum\limits_{N\leq k\leq -N}(-1)^{k}A_{k}A_{-k},\quad N<0,\\ &A_{0}^{2},\quad \quad \quad  N=0,\\ &0,\quad \quad \quad\ \  N>0.
\end{aligned}
\right. ,\quad \quad
Y^{2}=\left\{
\begin{aligned}
&\sum\limits_{M\leq l\leq -M}(-1)^{l}B_{l}B_{-l},\quad M<0,\\ &B_{0}^{2},\quad \quad \quad  M=0,\\ &0,\quad \quad \quad\ \  M>0.
\end{aligned}
\right.
\end{eqnarray}
\end{lemma}

\begin{lemma} \label{BUC solution jiaohuan}
\begin{eqnarray}
&&\Omega(X\otimes X)=(X\otimes X)\Omega,\quad \quad \quad \Omega(Y\otimes Y)=(Y\otimes Y)\Omega,\notag\\
&&\overline{\Omega}(X\otimes X)=(X\otimes X)\overline{\Omega},\quad \quad \quad \overline{\Omega}(Y\otimes Y)=(Y\otimes Y)\overline{\Omega}.
\end{eqnarray}
\end{lemma}

\begin{proof}
Using the similar approach in \cite{Kac2021}, we can prove the Lemma.
\end{proof}

\begin{corollary}
Let $\tau \in \mathcal{B}_{odd}$ be a tau-function of the BUC hierarchy, and let $X=\sum\limits_{n\geq N}A_{n}\varphi_{n}$, $Y=\sum\limits_{m\geq M}B_{m}\overline{\varphi}_{m}$, where $A_{n},B_{m}\in\mathbb{C}$ and $N,M\in\mathbb{Z}$. Then $\tau^{\prime}=X\tau$ and $\tau^{\prime\prime}=Y\tau$ are also tau-functions of the BUC hierarchy.
\end{corollary}

\begin{proof}
The proof method of this Corollary is similar to that of Corollary \ref{SKP corollary 2}, so it will not be described in detail here.
\end{proof}

\subsection{Generating functions and polynomial tau-functions of the BUC hierarchy}

Let $Q(\mathbf{u},\mathbf{v})$ be a generating function of the BUC hierarchy in $(\mathbf{u},\mathbf{v})=(u_{1},\ldots,u_{r},v_{1},\ldots,v_{s})$ defined by
\begin{eqnarray}\
&&Q(\mathbf{u},\mathbf{v})=\prod\limits_{1\leq i<j\leq r}\frac{u_i-u_j}{u_i+u_j}\prod\limits_{1\leq i<j\leq s}\frac{v_{i}-v_{j}}{v_{i}+v_{j}}\prod\limits_{{1\leq i\leq r}\atop{1\leq j\leq s}}\frac{1-u_{i}v_{j}}{1+u_{i}v_{j}}\prod\limits_{i=1}^{r}Q(u_{i})\prod\limits_{j=1}^{s}Q^{\prime}(v_{j}).
\end{eqnarray}
From Proposition \ref{BUC generating jiaohuan} and $S^{\bot}(u)(1)=1$, we have
\begin{eqnarray}\ \label{BUC feimi generating}
&&\varphi(u_{1})\cdots\varphi(u_{r})\overline{\varphi}(v_{1})\cdots\overline{\varphi}(v_{s})(1)=Q(\mathbf{u},\mathbf{v}).
\end{eqnarray}
It is expanded into rational function form $Q(\mathbf{u},\mathbf{v})=\sum\limits_{{\boldsymbol\alpha\in\mathbb{Z}^{r}}\atop{\boldsymbol\beta\in}\mathbb{Z}^{s}}Q_{\boldsymbol\alpha,\boldsymbol\beta}u_{1}^{\alpha_{1}}\cdots u_{r}^{\alpha_{r}}v_{1}^{\beta_{1}}\cdots v_{s}^{\beta_{s}}$, then
\begin{eqnarray}\
&&Q_{\boldsymbol\alpha,\boldsymbol\beta}=\varphi_{-\alpha_{1}}\cdots \varphi_{-\alpha_{r}}\overline{\varphi}_{-\beta_{1}}\cdots \overline{\varphi}_{-\beta_{s}}(1).
\end{eqnarray}

In the following, in order to express the family of generation functions of the BUC hierarchy as a certain Pfaffian, we denote the variables as $\check{\mathbf{u}}=_{def}(u_{1},u_{2},\ldots,u_{r}),\ \check{\mathbf{v}}=_{def}(u_{-1}^{-1},u_{-2}^{-1},\ldots,u_{-s}^{-1})$. Consider the set of  Laurent polynomial $A_{1}(u),\ldots,A_{l}(u)$ and $B_{-s}(u^{-1}),\ldots,B_{-1}(u^{-1})$, define a formal distribution
\begin{eqnarray}\ \label{T A B}
T(\check{\mathbf{u}},\check{\mathbf{v}})&=&\prod\limits_{j=1}^{r}A_{j}(u_{j})\prod\limits_{i=-s}^{-1}B_{i}(u_{i}^{-1})Q(\check{\mathbf{u}},\check{\mathbf{v}})\notag\\
&=&\prod\limits_{j=1}^{r}A_{j}(u_{j})\prod\limits_{i=-s}^{-1}B_{i}(u_{i}^{-1})\prod\limits_{{-s\leq i<j\leq r}\atop{i,j\neq0}}
f(u_{i},u_{j})\prod\limits_{i=-s}^{-1}Q^{\prime}(u_{i}^{-1})\prod\limits_{j=1}^{r}Q(u_{j}).
\end{eqnarray}
For any $\boldsymbol\gamma=(\gamma_{1},\ldots,\gamma_{r},\gamma_{-1},\ldots,\gamma_{-s})\in\mathbb{Z}^{r+s}$, $T_{\boldsymbol\gamma}$ is the coefficient of the following expansion
\begin{eqnarray}\ \label{BUC T r+s}
&&T(\check{\mathbf{u}},\check{\mathbf{v}})=\sum\limits_{\boldsymbol\gamma\in\mathbb{Z}^{r+s}}T_{\boldsymbol\gamma}u_{1}^{\gamma_{1}}\cdots u_{r}^{\gamma_{r}}u_{-1}^{\gamma_{-1}}\cdots u_{-s}^{\gamma_{-s}}.
\end{eqnarray}
We recall that if $A=[a_{ij}]$ is a skew symmetric matrix of even size $2n\times2n$, its determinant is a perfect square: $\det[A]=Pf[A]^{2}$, where
\begin{eqnarray}\
&&Pf[A]=\sum\limits_{\omega}sgn(\omega)a_{\omega(1)\omega(2)}\cdots a_{\omega(2n-1)\omega(2n)},
\end{eqnarray}
summed over $\omega\in S_{2n}$ such that $\omega(2r-1)<\omega(2r)$ for $1\leq r\leq n$, and $\omega(2r-1)<\omega(2r+1)$ for $1\leq r\leq n-1$. In addition, it is well-known that
\begin{eqnarray}\
&&Pf\left[\frac{u_{i}-u_{j}}{u_{i}+u_{j}}\right]_{1\leq i,j\leq2n}=\prod\limits_{1\leq i<j\leq2n}\frac{u_{i}-u_{j}}{u_{i}+u_{j}}.
\end{eqnarray}

Introduce the skew symmetric matrix $F=[f_{i,j}]_{{-2s\leq i,j\leq2r}\atop{i,j\neq0}}$, where
\begin{eqnarray}\
&&f_{i,j}=\left\{
\begin{aligned}
&f(u_{i},u_{j}),\quad\quad\quad i<j,\\ &0,\quad \quad \quad \quad\quad\quad i=j,\\ &-f(u_{j},u_{i}),\quad \ \ i>j.
\end{aligned}
\right.
\end{eqnarray}
Obviously,
\begin{eqnarray}\
&&Pf[F]=\prod\limits_{{-2s\leq i<j\leq2r}\atop{i,j\neq0}}f(u_{i},u_{j})=\prod\limits_{{-2s\leq i<j\leq2r}\atop{i,j\neq0}}\frac{u_{i}-u_{j}}{u_{i}+u_{j}}.
\end{eqnarray}
Define the formal distributions
\begin{eqnarray}\
&&T^{(i)}(u_{i}^{-1})=B_{i}(u_{i}^{-1})Q^{\prime}(u_{i}^{-1}),\quad i\in\{-2s,\ldots,-1\},\notag\\
&&T^{(j)}(u_{j})=A_{j}(u_{j})Q(u_{j}),\quad \quad \quad \ \ j\in\{1,\ldots,2r\},
\end{eqnarray}
and
\begin{eqnarray}\
&&T^{(i,j)}=f_{ij}T^{(i)}T^{(j)}=\sum\limits_{m,n\in\mathbb{Z}}T_{m,n}^{(i,j)}u_{i}^{m}u_{j}^{n},
\end{eqnarray}
where $T^{(i)}$ denotes $T^{(i)}(u_{i}^{-1})$ for negative $i$ and $T^{{i}}(u_{i})$ for positive $i$.

\begin{theorem} \label{BUC theorem}
\begin{itemize}
  \item [1)]The formal distribution $T(u_{1},\ldots,u_{2r},u_{-1}^{-1},\ldots,u_{-2s}^{-1})$ can be expressed as
  \begin{eqnarray}\ \label{BUC T even}
  &&T(u_{1},\ldots,u_{2r},u_{-1}^{-1},\ldots,u_{-2s}^{-1})=Pf[T^{(i,j)}]_{{-2s\leq i,j\leq2r}\atop{i,j\neq0}}.
  \end{eqnarray}
  \item [2)]
  The coefficient $T_{\check{\boldsymbol\gamma}}$ about the expansion in Eq.(\ref{BUC T even}) can be written as
  \begin{eqnarray}\
  &&T_{\breve{\boldsymbol\gamma}}=Pf[T^{(i,j)}_{\gamma_{i},\gamma_{j}}]_{{-2s\leq i,j\leq2r}\atop{i,j\neq0}},
  \end{eqnarray}
  where $\check{\boldsymbol\gamma}=(\gamma_{1},\ldots,\gamma_{2r},\gamma_{-1},\ldots,\gamma_{-2s})\in\mathbb{Z}^{2r+2s}$.
  \item [3)]
  For any $\boldsymbol\gamma=(\gamma_{1},\ldots,\gamma_{r},\gamma_{-1},\ldots,\gamma_{-s})\in\mathbb{Z}^{r+s}$, the coefficient $T_{\boldsymbol\gamma}$ about $u_{1}^{\gamma_{1}}\cdots u_{r}^{\gamma_{r}}u_{-1}^{\gamma_{-1}}\cdots u_{-s}^{\gamma_{-s}}$ in Eq.(\ref{BUC T r+s}) is a polynomial tau-function of the BUC hierarchy.
  \item [4)]
  There is a set of Laurent polynomials $A_{1}(u),\ldots,A_{r}(u),B_{-s}(u^{-1}),\ldots,B_{-1}(u^{-1})$ such that $\tau$ is the zero-mode of the Eq.(\ref{BUC T r+s}) if $\tau$ is a polynomial tau-function of the BUC hierarchy.
\end{itemize}
\end{theorem}

\begin{proof}
\begin{itemize}
  \item [1)]
  A direct calculation gives rise to
  \begin{eqnarray}\
  &&T(u_{1},u_{2},\ldots,u_{2r},u_{-1}^{-1},u_{-2}^{-1},\ldots,u_{-2s}^{-1})\notag\\
  &=&\prod\limits_{j=1}^{2r}A_{j}(u_{j})\prod\limits_{i=-2s}^{-1}B_{i}(u_{i}^{-1})\prod\limits_{{-2s\leq i<j\leq 2r}\atop{i,j\neq0}}
  f(u_{i},u_{j})\prod\limits_{i=-2s}^{-1}Q^{\prime}(u_{i}^{-1})\prod\limits_{j=1}^{2r}Q(u_{j})\notag\\
  &=&Pf[F]\prod\limits_{i=-2s}^{-1}T^{(i)}(u_{i}^{-1})\prod\limits_{j=1}^{2r}T^{(j)}(u_{j})\notag\\
  &=&\sum\limits_{\sigma\in S_{2s+2r}}sgn(\sigma)f_{\sigma(-2s)\sigma(-2s+1)}\cdots f_{\sigma(-2)\sigma(-1)}f_{\sigma(1)\sigma(2)}\cdots f_{\sigma(2r-1)\sigma(2r)}\cdot\notag\\
  &&\prod\limits_{i=-2s}^{-1}T^{(i)}(u_{i}^{-1})\prod\limits_{j=1}^{2r}T^{(j)}(u_{j})\notag\\
  &=&\sum\limits_{\sigma\in S_{2s+2r}}sgn(\sigma)f_{\sigma(-2s)\sigma(-2s+1)}T^{(\sigma(-2s))}T^{(\sigma(-2s+1))}\cdots f_{\sigma(-2)\sigma(-1)}T^{(\sigma(-2))}T^{(\sigma(-1))}\notag\\
  &&f_{\sigma(1)\sigma(2)}T^{(\sigma(1))}T^{(\sigma(2))}\cdots f_{\sigma(2r-1)\sigma(2r)}T^{(\sigma(2r-1))}T^{(\sigma(2r))}\notag\\
  &=&Pf[f_{ij}T^{(i)}T^{(j)}]_{{-2s\leq i,j\leq2r}\atop{i,j\neq0}}=Pf[T^{(i,j)}]_{{-2s\leq i,j\leq2r}\atop{i,j\neq0}}.
  \end{eqnarray}
  \item [2)]
  By the definition of $T^{(i,j)}$, after a straightforward calculation, we obtain
  \begin{eqnarray}\
  &&Pf[T^{(i,j)}]_{{-2s\leq i,j\leq2r}\atop{i,j\neq0}}=Pf\left[\sum\limits_{\gamma_{i},\gamma_{j}}T^{(i,j)}_{\gamma_{i},\gamma_{j}}u_{i}^{\gamma_{i}}
  u_{j}^{\gamma_{j}}\right]_{{-2s\leq i,j\leq2r}\atop{i,j\neq0}}\notag\\
  &=&\sum\limits_{\sigma\in S_{2s+2r}}sgn(\sigma)\sum\limits_{\gamma_{i},\gamma_{j}}T^{(\sigma(-2s),\sigma(-2s+1))}_{\gamma_{\sigma(-2s)},\gamma_{\sigma(-2s+1)}}
  u_{\sigma(-2s)}^{\gamma_{\sigma(-2s)}}u_{\sigma(-2s+1)}^{\gamma_{\sigma(-2s+1)}}\cdots T^{(\sigma(-2),\sigma(-1))}_{\gamma_{\sigma(-2)},\gamma_{\sigma(-1)}}
  u_{\sigma(-2)}^{\gamma_{\sigma(-2)}}u_{\sigma(-1)}^{\gamma_{\sigma(-1)}}\notag\\
  &&T^{(\sigma(1),\sigma(2))}_{\gamma_{\sigma(1)},\gamma_{\sigma(2)}}
  u_{\sigma(1)}^{\gamma_{\sigma(1)}}u_{\sigma(2)}^{\gamma_{\sigma(2)}}\cdots T^{(\sigma(2r-1),\sigma(2r))}_{\gamma_{\sigma(2r-1)},\gamma_{\sigma(2r)}}
  u_{\sigma(2r-1)}^{\gamma_{\sigma(2r-1)}}u_{\sigma(2r)}^{\gamma_{\sigma(2r)}}\notag\\
  &=&\sum\limits_{\boldsymbol\gamma}Pf[T^{(i,j)}_{\gamma_{i},\gamma_{j}}]_{{-2s\leq i,j\leq2r}\atop{i,j\neq0}}u_{1}^{\gamma_{1}}\cdots u_{2r}^{\gamma_{2r}}u_{-1}^{\gamma_{-1}}\cdots u_{-2s}^{\gamma_{-2s}}.
  \end{eqnarray}
  Clearly, the coefficient of $u_{1}^{\gamma_{1}}\cdots u_{2r}^{\gamma_{2r}}u_{-1}^{\gamma_{-1}}\cdots u_{-2s}^{\gamma_{-2s}}$ in Eq.(\ref{BUC T even}) is $Pf[T^{(i,j)}_{\gamma_{i},\gamma_{j}}]_{{-2s\leq i,j\leq2r}\atop{i,j\neq0}}$.
  \item [3)]
  Let $A_{j}(u)=\sum\limits_{M_{j}\leq k\leq N_{j}}A_{j,k}u^{j}$ and $B_{i}(u^{-1})=\sum\limits_{U_{i}\leq m \leq V_{i}}B_{i,m}u^{-m}$ be power series expansions about variable $u$, where $j=1,\ldots,l,i=-s,\ldots,-1$ and $M_{j},N_{j},U_{i},V_{i}\in\mathbb{Z}$. From (\ref{BUC feimi generating}), we can get
   \begin{align}\
   &T(\check{\mathbf{u}},\check{\mathbf{v}})=\prod\limits_{j=1}^{r}A_{j}(u_{j})\prod\limits_{i=-s}^{-1}B_{i}(u_{i}^{-1})\varphi(u_{1})\cdots\varphi(u_{r})
   \overline{\varphi}(u_{-1}^{-1})\cdots \overline{\varphi}(u_{-s}^{-1})(1)\notag\\
   &=\sum\limits_{M_{1}\leq k_{1}\leq N_{1}}\sum\limits_{l_{1}\in\mathbb{Z}}A_{1,k_{1}}\varphi_{l_{1}}u_{1}^{k_{1}-l_{1}}\cdots
   \sum\limits_{M_{r}\leq k_{r}\leq N_{r}}\sum\limits_{l_{r}\in\mathbb{Z}}A_{r,k_{r}}\varphi_{l_{r}}u_{r}^{k_{r}-l_{r}}
   \sum\limits_{U_{-1}\leq m_{-1}\leq V_{-1}}\sum\limits_{n_{-1}\in\mathbb{Z}}B_{-1,m_{-1}}\notag\\
   &\overline{\varphi}_{n_{-1}}u_{-1}^{n_{-1}-m_{-1}}
   \cdots\sum\limits_{U_{-s}\leq m_{-s}\leq V_{-s}}\sum\limits_{n_{-s}\in\mathbb{Z}}B_{-s,m_{-s}}\overline{\varphi}_{n_{-s}}u_{-s}^{n_{-s}-m_{-s}}(1)\notag\\
   &=\sum\limits_{\gamma\in\mathbb{Z}^{r+s}}\sum\limits_{M_{1}-\gamma_{1}\leq l_{1}\leq N_{1}-\gamma_{1}}A_{1,\gamma_{1}+l_{1}}
   \varphi_{l_{1}}u_{1}^{\gamma_{1}}\cdots\sum\limits_{M_{r}-\gamma_{r}\leq l_{r}\leq N_{r}-\gamma_{r}}A_{r,\gamma_{r}+l_{r}}
   \varphi_{l_{r}}u_{r}^{\gamma_{r}}\notag\\
   &\sum\limits_{U_{-1}+\gamma_{-1}\leq n_{-1}\leq V_{-1}+\gamma_{-1}}B_{-1,n_{-1}-\gamma_{-1}}\overline{\varphi}_{n_{-1}}u_{-1}^{\gamma_{-1}}\cdots\sum\limits_{U_{-s}+\gamma_{-s}\leq n_{-s}\leq V_{-s}+\gamma_{-s}}B_{-s,n_{-s}-\gamma_{-s}}\overline{\varphi}_{n_{-s}}u_{-s}^{\gamma_{-s}}(1).\notag\\
   \end{align}
Thus the coefficient $T_{\boldsymbol\gamma}$ of $u_{1}^{\gamma_{1}}\cdots u_{r}^{\gamma_{r}}u_{-1}^{\gamma_{-1}}\cdots u_{-s}^{\gamma_{-s}}$ can be written as $X_{1}\cdots X_{r}Y_{-1}\cdots Y_{-s}(1)$, where
  \begin{eqnarray}\ \label{BUC dengjia 1}
  &&X_{i}=\sum\limits_{M_{i}-\gamma_{i}\leq l_{i}\leq N_{i}-\gamma_{i}}A_{i,\gamma_{i}+l_{i}}\varphi_{l_{i}},\quad i=1,\ldots,r, \notag\\
  &&Y_{i}=\sum\limits_{U_{i}+\gamma_{i}\leq n_{i}\leq V_{i}+\gamma_{i}}B_{i,n_{i}-\gamma_{i}}\overline{\varphi}_{n_{i}},\quad i=-s,\ldots,-1.
  \end{eqnarray}
  By Remark \ref{BUC 1 solution} and Lemma \ref{BUC solution jiaohuan}, we conclude that the coefficient $T_{\boldsymbol\gamma}$ is a tau-function of the BUC hierarchy. $T_{\boldsymbol\gamma}$ is a polynomial tau-function because it is a finite linear combination of  $\varphi_{l_{1}}\cdots\varphi_{l_{r}}\overline{\varphi}_{n_{-1}}\cdots\overline{\varphi}_{n_{-s}}(1)$.
  \item [4)]
  Polynomial tau-function of the BUC hierarchy has the form
  \begin{eqnarray} \label{given BUC tau}
  &&\tau=X_{1}\cdots X_{r}Y_{-1}\cdots Y_{-s}(1),
  \end{eqnarray}
  where $[\lambda,\mu]=[(\lambda_{1},\lambda_{2},\ldots,\lambda_{r}),(-\lambda_{-1},-\lambda_{-2},\ldots,-\lambda_{-s})]$ is a pair of partitions, and
  \begin{eqnarray} \label{BUC dengjia 2}
  &&X_{i}=\sum\limits_{-\lambda_{i}\leq m\leq N_{i}}d_{m,i}\varphi_{m},\quad d_{m,i}\in\mathbb{C},b_{-\lambda_{i},i}\neq0,N_{i}\in\mathbb{Z},i=1,\ldots,r,\notag\\
  &&Y_{i}=\sum\limits_{-\lambda_{i}\leq h\leq V_{i}}e_{h,i}\overline{\varphi}_{h},\quad e_{h,i}\in\mathbb{C},e_{-\lambda_{i},i}\neq0,V_{i}\in\mathbb{Z},i=-s,\ldots,-1.
  \end{eqnarray}
  For a  vector $\boldsymbol\gamma=(\gamma_{1},\ldots,\gamma_{r},\gamma_{-1},\ldots,\gamma_{-s})\in\mathbb{Z}^{r+s}$, we define $A_{i}(u)$ and $B_{i}(u^{-1})$ in the Eq.(\ref{T A B}) to be the Laurent polynomial with the following form
  \begin{eqnarray}
  &&A_{i}(u)=\sum\limits_{\gamma_{i}-\lambda_{i}\leq t\leq N_{i}+\gamma_{i}}d_{t-\gamma_{i},i}u^{t},\quad \quad \quad \quad i=1,\ldots,r,\notag\\
  &&B_{i}(u^{-1})=\sum\limits_{-\gamma_{i}-\lambda_{i}\leq w\leq V_{i}-\gamma_{i}}e_{w+\gamma_{i},i}u^{-w},\quad i=-s,\ldots,-1.
  \end{eqnarray}
  By using $A_{1}(u),\ldots,A_{r}(u)$ and $B_{-s}(u^{-1}),\ldots,B_{-1}(u^{-1})$, it is easy to verify that Eq. (\ref{BUC dengjia 1}) leads to (\ref{BUC dengjia 2}). The coefficient $T_{\boldsymbol\gamma}$ corresponds to the polynomial tau-function (\ref{given BUC tau}). It is showed that $\tau$ is the zero-mode of the series expansion of $T(\check{\mathbf{u}},\check{\mathbf{v}})$ with $\gamma_{1}=\cdots=\gamma_{r}=\gamma_{-1}=\cdots=\gamma_{-s}=0$ .
\end{itemize}
\end{proof}

\begin{corollary}
\begin{itemize}
  \item [1)]We have proved that the polynomial tau-functions of the BUC hierarchy are zero-mode of certain generating functions $T(\check{\mathbf{u}},\check{\mathbf{v}})$. By replacing  $A_{j}(u)$ with $u^{\gamma_{j}}A_{j}(u)\ (j=1,\ldots,r)$ and $B_{i}(u^{-1})$ with $u^{-\gamma_{i}}B_{i}(u^{-1})\ (i=-s,\ldots,-1)$, we derive the any polynomial tau-function as a coefficient of a given monomial $u_{1}^{\gamma_{1}}\cdots u_{r}^{\gamma_{r}}u_{-1}^{\gamma_{-1}}\cdots u_{-s}^{\gamma_{-s}}$.
  \item [2)]Introducing
  \begin{eqnarray}\
  &&A_{j}(u)=h_{j}\sum\limits_{i=0}^{M_{j}}a_{j,i}u^{i},\quad \quad \quad \ M_{j}\in\mathbb{Z},h_{j},a_{j,i}\in\mathbb{C},a_{j,0}=1,j=1,\ldots,r,\notag\\
   &&B_{i}(u^{-1})=g_{i}\sum\limits_{m=0}^{M_{i}}b_{i,m}u^{-m},\quad M_{i}\in\mathbb{Z},g_{i},b_{i,m}\in\mathbb{C},b_{i,0}=1,i=-s,\ldots,-1,
   \end{eqnarray}
   where $A_{1}(u),\ldots A_{r}(u),B_{-1}(u_{-1}),\ldots,B_{-s}(u_{-s})$ are non-zero Laurent series defined in the $T(\check{\mathbf{u}},\check{\mathbf{v}})$. By means of (\ref{Schur Taylor}), we have
   \begin{eqnarray}\
   &&\sum\limits_{i=0}^{M_{j}}a_{j,i}u^{i}=\exp\left(\sum\limits_{s=1}^{\infty}c_{j,s}u^{s}\right),\quad \quad
   \sum\limits_{m=0}^{M_{i}}b_{i,m}u^{-m}=\exp\left(\sum\limits_{l=1}^{\infty}c^{\prime}_{i,l}u^{-l}\right),
  \end{eqnarray}
  and
  \begin{eqnarray}\
  &&a_{j,i}=S_{i}(c_{j,1},c_{j,2},\ldots),\quad \quad b_{i,m}=S_{m}(c^{\prime}_{i,1},c^{\prime}_{i,2},\ldots),
  \end{eqnarray}
  where $c_{j,s}$ and $c^{\prime}_{i,l}$ are constants in $\mathbb{C}$.

  Setting $(\widetilde{x}_{1},\widetilde{x}_{2},\widetilde{x}_{3},\ldots)=(2p_{1},0,\frac{2}{3}p_{3},0,\ldots)$ and $(\widetilde{x}^{\prime}_{1},\widetilde{x}^{\prime}_{2},\widetilde{x}^{\prime}_{3},\ldots)=(2p^{\prime}_{1},0,\frac{2}{3}p^{\prime}_{3},0,\ldots)$, we get
  \begin{eqnarray}\
  T^{(j)}(u_{j})&=&A_{j}(u_{j})Q(u_{j})=h_{j}\exp\left(\sum\limits_{k\geq1}c_{j,k}u_{j}^{k}\right)\exp\left(\sum\limits_{k\geq1}
  \widetilde{x}_{k}u_{j}^{k}\right)\notag\\
  &=&h_{j}\sum\limits_{k=0}^{\infty}
  S_{k}(\widetilde{x}_{1}+c_{j,1},\widetilde{x}_{2}+c_{j,2},\ldots)u_{j}^{k},\quad j=1,\ldots,r,\notag\\
  T^{(i)}(u_{i}^{-1})&=&B_{i}(u_{i}^{-1})Q^{\prime}(u_{i}^{-1})=g_{i}\exp\left(\sum\limits_{l\geq1}c^{\prime}_{i,l}u_{i}^{-l}\right)
  \exp\left(\sum\limits_{l\geq1}\widetilde{x}^{\prime}_{l}u_{i}^{-l}\right)\notag\\
  &=&g_{i}\sum\limits_{l=0}^{\infty}S_{l}(\widetilde{x}^{\prime}_{1}+c^{\prime}_{i,1},\widetilde{x}^{\prime}_{2}+c^{\prime}_{i,2},
  \ldots)u_{i}^{-l},\quad i=-s,\ldots,-1.
  \end{eqnarray}
  Hence, $T^{(i,j)}$ can be written as
  \begin{align}\
  T^{(i,j)}=\left\{
  \begin{aligned}
  &h_{i}h_{j}\left(1+2\sum\limits_{k\geq1}(-1)^{k}\frac{u_{j}^{k}}{u_{i}^{k}}\right)\sum\limits_{m,n\in\mathbb{Z}}S_{m}
  \left(\widetilde{x}+c_{i}\right)S_{n}\left(\widetilde{x}+c_{j}\right)u_{i}^{m}u_{j}^{n}, \quad\quad \ 0<i<j\leq 2r,\\
  &g_{i}g_{j}\left(1+2\sum\limits_{k\geq1}(-1)^{k}\frac{u_{j}^{k}}{u_{i}^{k}}\right)\sum\limits_{m,n\in\mathbb{Z}}S_{m}
  \left(\widetilde{x}^{\prime}+c^{\prime}_{i}\right)S_{n}\left(\widetilde{x}^{\prime}+c^{\prime}_{j}\right)u_{i}^{-m}u_{j}^{-n},\quad
  -2s\leq i<j<0,\\
  &h_{i}g_{j}\left(1+2\sum\limits_{k\geq1}(-1)^{k}\frac{u_{j}^{k}}{u_{i}^{k}}\right)\sum\limits_{m,n\in\mathbb{Z}}S_{m}
  \left(\widetilde{x}^{\prime}+c^{\prime}_{i}\right)S_{n}\left(\widetilde{x}+c_{j}\right)u_{i}^{-m}u_{j}^{n},-2s\leq i<0<j\leq2r.
  \end{aligned}
  \right.
  \end{align}
  Expanding $T^{(i,j)}$, we can obtain the expression of the $T^{(i,j)}_{m,n}$
  \begin{align}
  T^{(i,j)}_{m,n}=\left\{
  \begin{aligned}
  &2h_{i}h_{j}\mathcal{X}^{(1)}_{m,n}\left(\widetilde{x}+c_{i},\widetilde{x}+c_{j}\right),\quad\quad \quad 0<i<j\leq 2r,\\
  &2g_{i}g_{j}\mathcal{X}^{(2)}_{m,n}\left(\widetilde{x}^{\prime}+c^{\prime}_{i},\widetilde{x}^{\prime}+c^{\prime}_{j}\right),\quad
  \quad -2s\leq i<j<0,\\
  &2h_{i}g_{j}\mathcal{X}^{(3)}_{m,n}\left(\widetilde{x}^{\prime}+c^{\prime}_{i},\widetilde{x}+c_{j}\right),\quad -2s\leq i<0<j\leq2r,
  \end{aligned}
  \right.
  \end{align}
  where  $T^{(i,j)}_{m,n}=-T^{(j,i)}_{m,n}$ for $i>j$, $T^{(i,i)}_{m,n}=0$, and
  \begin{align}
  &\mathcal{X}^{(1)}_{m,n}\left(\widetilde{x}+c_{i},\widetilde{x}+c_{j}\right)=\frac{1}{2}S_{m}(\widetilde{x}+c_{i})S_{n}(\widetilde{x}+c_{j})
  +\sum\limits_{k\geq1}(-1)^{k}S_{m+k}(\widetilde{x}+c_{i})S_{n-k}(\widetilde{x}+c_{j}),\notag\\
  &\mathcal{X}^{(2)}_{m,n}\left(\widetilde{x}^{\prime}+c^{\prime}_{i},\widetilde{x}^{\prime}+c^{\prime}_{j}\right)=\frac{1}{2}S_{m}
  (\widetilde{x}^{\prime}+c^{\prime}_{i})S_{n}(\widetilde{x}^{\prime}+c^{\prime}_{j})+\sum\limits_{k\geq1}(-1)^{k}S_{m-k}
  (\widetilde{x}^{\prime}+c^{\prime}_{i})S_{n+k}(\widetilde{x}^{\prime}+c^{\prime}_{j}),\notag\\
  &\mathcal{X}^{(3)}_{m,n}\left(\widetilde{x}^{\prime}+c^{\prime}_{i},\widetilde{x}+c_{j}\right)=\frac{1}{2}S_{m}
  (\widetilde{x}^{\prime}+c^{\prime}_{i})S_{n}(\widetilde{x}+c_{j})++\sum\limits_{k\geq1}(-1)^{k}S_{m-k}
  (\widetilde{x}^{\prime}+c^{\prime}_{i})S_{n-k}(\widetilde{x}+c_{j}).
  \end{align}
  In order to facilitate expression, some new symbols will be  introduced. Define the skew-symmetric matrix $M=(m_{i,j})_{0<i,j\leq 2r}$ by putting each $(i,j)$-th as $m_{i,j}=2h_{i}h_{j}\mathcal{X}^{(1)}_{m,n}\left(\widetilde{x}+c_{i},\widetilde{x}+c_{j}\right)$ for $i<j$ and $m_{i,j}=-m_{j,i}$ for $i>j$, $m_{i,i}=0$. Similarly, define another skew-symmetric matrix $\overline{M}=(\overline{m}_{i,j})_{-2s\leq i,j<0}$, where $\overline{m}_{i,j}=2g_{i}g_{j}\mathcal{X}^{(2)}_{m,n}\left(\widetilde{x}^{\prime}+c^{\prime}_{i},\widetilde{x}^{\prime}+c^{\prime}_{j}\right)$
  for $i<j$. Introduce the third matrix $N=(n_{i,j})_{-2s\leq i<0,0<j\leq2r}$, where $n_{i,j}=2h_{i}g_{j}\mathcal{X}^{(3)}_{m,n}\left(\widetilde{x}^{\prime}+c^{\prime}_{i},\widetilde{x}+c_{j}\right)$.\\
  From Theorem \ref{BUC theorem}, polynomial tau-functions of the BUC hierarchy have the form
  \begin{eqnarray}\
  &&T_{\boldsymbol\gamma}=Pf
  \begin{bmatrix}
  \overline{M} &  N \\
  -N^{T} & M
  \end{bmatrix}_{-2s\leq i,j\leq2r.}
  \end{eqnarray}
\end{itemize}
\end{corollary}

\begin{remark}
It is noted that the polynomial tau-functions of the BUC hierarchy reduce to the  solutions of the BKP hierarchy \cite{Kac2021} with the reduction $\mathbf{y}=0$.
\end{remark}

\section{Conclusions and discussions}

In this paper, we have discussed exact solutions of the SKP, OKP and BUC hierarchies including the polynomial-type and soliton-type solutions. It is showed that the generating functions  play a vital role in establishing the polynomial tau-functions of the integrable systems. Furthermore, we expressed the polynomial tau-functions of the SKP, OKP and BUC hierarchies as determinant and  Pfaffian forms, respectively. The results here are hoped to be helpful for better understanding the essential properties of the SKP, OKP and BUC hierarchies. It is known that symplectic universal character (SUC) and orthogonal universal character (OUC) hierarchies are the extensions of the SKP and OKP hierarchies. However, it should be pointed out that we have not expressed the polynomial tau-functions of the SUC and OUC hierarchies as a perfect determinant form due to the inappropriate quantum fields presentation of SUC and OUC. We will concentrate on studying this interesting question in the near future.

\section{Acknowledgements}
This work is partially supported by National Natural Science Foundation of China (Grant Nos. 11965014 and 12061051) and National Science Foundation of Qinghai Province (Grant No. 2021-ZJ-708).  The authors gratefully acknowledge the support of Professor Ke Wu and  Professor Weizhong Zhao (CNU, China).


\begin{thebibliography}{99}

\bibitem{Weyl1946} H. Weyl, The Classical Groups: Their Invariants and Representations, Princeton Univ. Press (1946).

\bibitem{Fulton1991} W. Fulton, J. Harris, Representation Theory: A First Course, Springer-Verlag, New York (1991).

\bibitem{Macdonald1995} I. G. Macdonald, Symmetric Functions and Hall Polynomials, Clarendon Press, Oxford (1995).


\bibitem{Jing1991} N. Jing, Vertex operators, symmetric functions, and the spin group $\Gamma_n$, J. Algebra 138 (1991) 340-398.


\bibitem{Jimbo1}M. Jimbo, T. Miwa and E. Date, Solitons: Differential Equations, Symmetries and Infinite Dimensional Algebras, Cambridge Univ. Press, Cambridge (2000).

\bibitem{BEsymmetricgroup2013}B. E. Sagan, The Symmetric Group: Representations, Combinatorial Algorithms, and Symmetric Functions, Springer Science and Business Media (2013).

\bibitem{Sato1981}M. Sato, Soliton equations as dynamical systems on a infinite-dimensional Grassmann manifold, Res. Inst. Math. Sci. Kokyuroku 439 (1981) 30-46.

\bibitem{J1}E. Date, M. Jimbo, M. Kashiwara and T. Miwa, Operator approach to the Kadomtsev-Petviashvili equation-transformation groups for soliton equations III, J. Phys. Soc. Jpn. 50 (1981) 3806-3812.

\bibitem{J2}E. Date, M. Jimbo, M. Kashiwara and T. Miwa, Transformation groups for soliton equations: IV. A new hierarchy of soliton equations of KP-type, Physica D 4 (1982) 343-365.

\bibitem{Jimbo3}E. Date, M. Jimbo, M. Kashiwara and T. Miwa, KP hierarchies of orthogonal and symplectic type-transformation groups for soliton equations VI, J. Phys. Soc. Jpn. 50 (1982) 3813-3818.

\bibitem{J3}E. Date, M. Jimbo, M. Kashiwara and T. Miwa, Transformation groups for soliton equations-Euclidean Lie algebras and reduction of the KP hierarchy, Publ. Res. Inst. Math. Sci. 18 (1982) 1077-1110.

\bibitem{J4}M. Jimbo and T. Miwa, Solitons and infinite-dimensional Lie algebras, Publ. Res. Inst. Math. Sci. 19 (1983) 943-1001.

\bibitem{J5}E. Date, M. Kashiwara, M. Jimbo and T. Miwa, Transformation groups for soliton equations, in  ``Nonlinear integrable systems-classical theory and quantum theory",  World Scientific Publishing, Singapore (1983) 39-119.

\bibitem{Koike1989}K. Koike, On the decomposition of tensor products of the representations of the classical groups: by means of the universal characters, Adv. Math. 74 (1989) 57-86.

\bibitem{Tsuda2004}T. Tsuda, Universal characters and an extension of the KP hierarchy, Commun. Math. Phys. 248 (2004) 501-526.


\bibitem{Tsuda2005}T. Tsuda, Universal characters, integrable chains and the Painlev\'{e} equations, Adv. Math. 197 (2005) 587-606.

\bibitem{Tsuda2005 2} T. Tsuda, Universal characters and $q$-Painlev\'{e} systems, Comm. Math. Phys. 260 (2005) 59-73.

\bibitem{Tsuda2009}T. Tsuda, Universal character and $q$-difference Painlev\'{e} equations, Math. Ann. 345 (2009) 395-415.

\bibitem{Tsuda2012}T. Tsuda, From KP/UC hierarchies to Painlev\'{e} equations, Int. J. Math. 23 (2012) 1250010.


\bibitem{Wang2019}N. Wang and C.Z. Li, Universal character, phase model and topological strings on $\mathbb{C}^{3}$, Eur. Phys. J. C 79 (2019) 1-9.


\bibitem{Baker1996}T. H. Baker, Vertex operator realization of symplectic and orthogonal $S$-functions, J. Phys. A 29 (1996) 3099-3117.



\bibitem{Wang2020}F. Huang and  N. Wang, Generalized symplectic Schur functions and SUC hierarchy, J. Math. Phys. 61  (2020) 061508.

\bibitem{Wang2021}L. J. Shi, N. Wang and M. R. Chen, The orthogonal and symplectic Schur functions, vertex operators and integrable hierarchies, J. Nonlinear Math. Phys. 28 (2021) 292-302.


\bibitem{Ogawa2009}Y. Ogawa, Generalized Q-functions and UC hierarchy of B-type, Tokyo J. Math. 32 (2009) 350-380.


\bibitem{Li2019}C. Z. Li, Strongly coupled B-type universal characters and hierarchies,  Theor. Math. Phys. 201 (2019) 1732-1741.

\bibitem{Li2021}C. Z. Li, Plethystic B-type KP and universal character hierarchies, J.  Algebr.  Comb.  DOI: 10.1007/s10801-021-01066-2  (2021).

\bibitem{You1989}Y. You, Polynomial solutions of the BKP hierarchy and projective representations of symmetric groups, in infinite-dimensional Lie algebras and groups (Luminy-Marseille, 1988) Adv. Ser. Math. Phys. 7 (1989) 449-464.

\bibitem{You1991}Y. You, DKP and MDKP hierarchy of soliton equations, Physica D 50 (1991) 429-462.

\bibitem{Kac2018}V. G. Kac and J. W. van de Leur, Equivalence of formulations of the MKP hierarchy and its polynomial tau-functions, Jpn. J. Math. 13 (2018) 235-271.

\bibitem{Kac2019}V. G. Kac and J. W. van de Leur, Polynomial tau-functions of BKP and DKP hierarchies, J. Math. Phys. 60 (2019) 071702.

\bibitem{Kac2019M}V. G. Kac and J. W. van de Leur, Polynomial tau-functions for the multi-component KP hierarchy, arXiv:1901.07763.

\bibitem{Vecochea2020}G. Necoechea and N. Rozhkovskaya, Generalized vertex operators of Hall--Littlewood polynomials as twists of charged free fermions, J. Math. Sci. 247 (2020) 926-938.

\bibitem{Rozhkovskaya2019}N. Rozhkovskaya, Multiparameter Schur Q-functions sre solutions of the BKP hierarchy, Symmetry
     Integrability Geom. Methods Appl.  15 (2019) 065.

\bibitem{Kac2021}V. G. Kac, N. Rozhkovskaya  and  J. W. van de Leur,  Polynomial tau-functions of the KP, BKP, and the $s$-component KP hierarchies, J. Math. Phys. 62 (2021) 021702.

\bibitem{LiCZ2022}C. Z. Li, Multi-component universal character hierarchy and its polynomial tau-functions, Physica D DOI: 10.1016/j.physd.2022.133166 (2022).


\bibitem{Kac2013} V. G. Kac, A. K. Raina and  N. Rozhkovskaya, Bombay Lectures on Highest Weight Representations of Infinite Dimensional Lie Algebras, World Scientific, Hackensack, NJ (2013).

\bibitem{Li2022} D. H. Li, F. Wang and Z. W. Yan, Quantum fields presentation and generating functions of symplectic Schur functions and symplectic universal characters, Chinese Phys. B  DOI: 10.1088/1674-1056/ac4f57 (2022).


\bibitem{Jing2015} N. Jing and B. Z. Nie, Vertex operators, Weyl determinant formulae and Littlewood duality, Ann. Comb. 19 (2015) 427-442.


\end{thebibliography}
\end{document}